\documentclass[2pt]{article}
\usepackage[intlimits]{amsmath}
\usepackage{amssymb}
\usepackage{amsthm}
\usepackage[T1]{fontenc}
\usepackage[utf8]{inputenc}

\usepackage[text={15.4cm,22cm},centering]{geometry}

\usepackage{expdlist}
\usepackage{shortlst}
\usepackage{enumerate}
\usepackage{bm}
\usepackage[footnote,marginclue]{fixme} 
\usepackage[all]{xy}

\usepackage{pifont}

\newtheorem{theorem}{Theorem}[section]
\newtheorem{lemma}[theorem]{Lemma}
\newtheorem{corollary}[theorem]{Corollary}
\newtheorem{proposition}[theorem]{Proposition}

\theoremstyle{definition}
\newtheorem{definition}[theorem]{Definition}
\newtheorem{example}[theorem]{Example}

\DeclareBoldMathCommand\vo{0}
\DeclareBoldMathCommand\va{a}
\DeclareBoldMathCommand\vb{b}
\DeclareBoldMathCommand\vc{c}
\DeclareBoldMathCommand\vd{d}
\DeclareBoldMathCommand\ve{e}
\DeclareBoldMathCommand\vn{n} 
\DeclareBoldMathCommand\vp{p}   
\DeclareBoldMathCommand\vq{q}   
\DeclareBoldMathCommand\vs{s}   
\DeclareBoldMathCommand\vt{t}
\DeclareBoldMathCommand\vu{u}
\DeclareBoldMathCommand\vx{x}
\DeclareBoldMathCommand\vy{y}
\DeclareBoldMathCommand\vz{z}

\DeclareBoldMathCommand{\vU}{U}
\DeclareBoldMathCommand{\vV}{V}
\DeclareBoldMathCommand{\vW}{W}
\DeclareBoldMathCommand{\vX}{X}
\DeclareBoldMathCommand{\vY}{Y}
\DeclareBoldMathCommand{\vZ}{Z}
\DeclareBoldMathCommand{\vk}{\kappa}
\DeclareBoldMathCommand{\vl}{\lambda}
\DeclareBoldMathCommand{\vm}{\mu}

\newcommand{\BR}{\mathrm{br}}

\newcommand{\e}{\varepsilon}
\newcommand{\g}{\gamma}
\newcommand{\kh}{\chi}
\newcommand{\oo}{\omega}

\newcommand{\B}{\mathcal{B}}
\newcommand{\A}{\mathcal{A}}
\newcommand{\F}{\mathcal{F}}
\newcommand{\LL}{\mathcal{L}}

\newcommand{\mA}{{\mathfrak A}}
\newcommand{\mM}{{\mathfrak M}}
\newcommand{\mN}{{\mathfrak N}}

\newcommand{\NN}{\mathbb{N}}
\newcommand{\ZZ}{\mathbb{Z}}
\newcommand{\RR}{\mathbb{R}}
\newcommand{\PP}{\mathbb{P}}

\DeclareMathOperator{\kC}{{\sf C}}
\DeclareMathOperator{\kI}{{\sf I}}
\DeclareMathOperator{\kQ}{{\sf Q}}

\DeclareMathOperator{\kD}{{\sf D}}

\DeclareMathOperator{\Maj}{{\sf Maj}}
\DeclareMathOperator{\kB}{{\sf B}}
\newcommand{\Q}{\mathcal{Q}}

\DeclareMathOperator{\FO}{FO}
\DeclareMathOperator{\MSO}{MSO}
\DeclareMathOperator{\FOC}{FOC}
\DeclareMathOperator{\RI}{\mathcal{R}-int}
\DeclareMathOperator{\RC}{\mathcal{R}-cl}

\newcommand{\reg}{\mathrm{reg}}

\DeclareMathOperator{\Str}{Str}
\DeclareMathOperator{\Dom}{Dom}
\DeclareMathOperator{\dom}{dom}

\newcommand{\TC}{\mathrm{TC^0}}
\newcommand{\AC}{\mathrm{AC^0}}

\newcommand{\LFP}{\mathrm{LFP}}

\newcommand{\NP}{\mathrm{NP}}
\newcommand{\PTIME}{\mathrm{PTIME}}
\newcommand{\DLOGTIME}{\mathrm{DLOGTIME}}
\newcommand{\PSPACE}{\mathrm{PSPACE}}
\newcommand{\LOGSPACE}{\mathrm{LOGSPACE}}
\newcommand{\NLOGSPACE}{\mathrm{NLOGSPACE}}

\newcommand{\osaj}{\subseteq}

\newcommand{\ekv}{\leftrightarrow}

\DeclareMathOperator{\arity}{ar}
\DeclareMathOperator{\card}{card}

\DeclareMathOperator{\rg}{rg}
\DeclareMathOperator{\lb}{lb}

\newcommand{\Sq}{\mathrm{Sq}}

\newcommand{\bplus}{\text{\ding{58}}}
\newcommand{\btimes}{\text{\ding{54}}}

\newcommand{\bit}{\mathrm{BIT}}

\newcommand{\katto}[1]{\left\lceil #1\right\rceil}
\newcommand{\lattia}[1]{\left\lfloor #1\right\rfloor}
\newcommand{\idiv}[2]{\lattia{\frac{#1}{#2}}}
\newcommand{\joukko}[1]{\left\{ #1\right\}}

\newcommand{\Pmod}[1]{{\!\!\pmod{#1}}}

\newcommand{\sij}{\mathrel{\mathop:}=}
\newcommand{\raj}{\restriction}
\newcommand{\av}[1]{{]}#1{[}}
\newcommand{\sav}[1]{{[}#1{[}}

\newcommand{\viiva}{\mathop{\Big/}}
\newcommand{\pri}[2]{\viiva\limits_{\!\!\!\!{#1}}^{\>\,{#2}}}
\newcommand{\dt}{\;\textrm{d}t}

\newcommand{\pad}{{\le_{\mathrm P}}}

\title{Regular Representations of Uniform $\TC$
   \thanks{The second author was supported by grants 127661, 264917, and 345634  
           of the Academy of  Finland. 
           The authors are grateful for the generous support 
           of the Mittag-Leffler Institute, 
           where part of the work reported here was carried out.}}

\author{Lauri Hella 
   \and Juha Kontinen 
      \and Kerkko Luosto
}

\date{}

\begin{document}

\maketitle

\begin{abstract}

The circuit complexity class $\DLOGTIME$-uniform $\AC$ is known to be
a modest subclass of $\DLOGTIME$-uniform $\TC$. The weakness of $\AC$
is caused by the fact that $\AC$ is not closed under restricting
$\AC$-computable queries into simple subsequences of the
input. Analogously, in descriptive complexity, the logics
corresponding to $\DLOGTIME$-uniform $\AC$ do not have the
relativization property and hence they are not regular. This weakness
of $\DLOGTIME$-uniform $\AC$ has been elaborated in the line of
research on the Crane Beach Conjecture. The conjecture (which was
refuted by Barrington, Immerman, Lautemann, Schweikardt 
and Th{\'e}rien in
\cite{BILST}) was that if a language $L$ has a neutral
letter, then $L$ can be defined in $\FO_\A$, first-order logic with
the collection of all numerical built-in relations~$\A$, if and only if 
$L$ can be already defined in $\FO_{\le}$.

In the first part of this article we consider logics in the range of
$\AC$ and~$\TC$. First we formulate a combinatorial criterion for a
cardinality quantifier~$\kC_S$ implying that all languages in
$\DLOGTIME$-uniform $\TC$ can be defined in $\FO_{\le}(\kC_S)$.  For
instance, this criterion is satisfied by $\kC_S$ if $S$ is the range
of some polynomial with positive integer coefficients of degree at
least two.  In the second part of the paper we first adapt the key
properties of abstract logics to accommodate built-in relations. Then
we define the regular interior $\RI(\LL)$ and regular closure
$\RC(\LL)$, of a logic~$\LL$, and show that the Crane Beach Conjecture
can be interpreted as a statement concerning $\RI(\FO_\B)$. By
extending the results of \cite{BILST}, we show that if $\B=\{ +\}$, or
$\B$ contains only unary relations besides $\le$, then
$\RI(\FO_\B)\equiv\FO_{\le}$. In contrast, our results imply that if
$\B$ contains $\le$ and the range of a polynomial of degree at least
two, then $\RC(\FO_\B)$ includes all languages in 
$\DLOGTIME$-uniform~$\TC$.

\end{abstract}

\section{Introduction}


Many circuit complexity classes have been logically characterized by
extensions of first-order logic in terms of varying sets of built-in
relations and generalized quantifiers. The seminal paper in this area
is \cite{BIS} in which the connection between $\DLOGTIME$-uniformity
and $\FO_{\{+,\times\}}$-definability was established, where
$\FO_{\{+,\times\}}$ denotes first-order logic with ternary built-in
relations $+$ and $\times$. It was also shown in \cite{BIS} that the
languages in $\DLOGTIME$-uniform $\AC$ are exactly those definable in
$\FO_{\{+,\times\}}$. It is known that the predicates $+$ and $\times$
can be defined in terms of the $\bit$ relation. In fact it was shown
in \cite{DDLW} that $\bit$ alone can define the corresponding
canonical ordering, hence $\FO_{\{+,\times\}}\equiv \FO_{\bit}$ (in
the case of a single built-in relation, we drop the set parenthesis in
the subscript).  A provably larger class of languages, $\TC$, is
acquired by allowing also majority gates in the circuits. On the
logical side, $\DLOGTIME$-uniform $\TC$ corresponds to the extension
$\FO_{\{+,\times\}}(\Maj)$ of first-order logic by the unary majority
quantifier $\Maj$ with built-in $+$ and $\times$. We refer later to
$\DLOGTIME$-uniform $\AC$ and DLOGTIME-uniform $\TC$ simply as $\AC$
and $\TC$. It was shown in \cite{BIS} that the "majority of pairs"
$\Maj^2$ can be expressed in terms of $\Maj$ with the help of the
numerical relations. On the other hand, with $\Maj^2$ and order, the
numerical relations become definable. In \cite{BIS} it was also asked
whether already $\Maj$ is enough to define the numerical
relations. This has been shown not to hold in \cite{Li}, \cite{R} and
\cite{LMcKSV}. On the other hand, in \cite{Lu1}, it was observed that
the extension of $\FO_{\le}$ by the general divisibility
quantifier~$\kD$ is enough to capture $\TC$. Note that, unlike
$\Maj^2$, $\kD$~is a unary quantifier. In \cite{HAB}, $\TC$ was shown
to include the problems of division and iterated multiplication of
binary numbers.  These results make essential use of the logical
characterization of $\TC$ discussed above.

It is interesting to note that, while even non-uniform $\AC$ fails to
define \emph{Parity} by the famous Theorem of Ajtai \cite{A} and
Furst, Saxe and Sipser \cite{FSS}, it is not known if $\NP$ strictly
includes $\TC$. The weakness of $\AC$ is due to the fact that $\AC$ is
not closed under restricting $\AC$-computable queries into  simple 
 subsequences of the input, e.g., the query $Q$ expressing that the length of a binary word $v$
is odd is $\AC$-computable, but, with input $v$, no $\AC$-query can simulate $Q$ over the  subsequence $v'$ of $v$
arising by deleting all $0$'s in $v$. On the logical side, the
logics corresponding to $\AC$ do not have the so-called
\emph{relativization property} and hence they are not \emph{regular}
(for the definition of these concepts, see \cite{E} and Section
\ref{Builtin}).  In fact, we will show that the only thing that $\AC$
lacks, compared to $\TC$, is the ability to relativize. The weakness
of $\AC$, and logics $\FO_\B$, where $\B$ is a collection of built-in
relations, is also reflected in the fact that $\FO_\B$ cannot count
cardinalities of sets.  This defect can be fixed by extending $\FO_\B$ (${\le} \in \B$ in the following discussion)
in terms of counting quantifiers $\exists ^{=y} x$, resulting with the
logic $\FOC_\B$ (see, e.g., \cite{S2}). By the assumption ${\le} \in \B$, we can
equivalently extend $\FO_\B$ by $\Maj$ or the H\"artig quantifier
$\kI$ expressing equicardinality \cite{Lu1}.  In the presence of $\kI$
(respectively, $\Maj$ or $\exists ^{=y} x$) built-in relations can be
replaced by certain (unary) generalized quantifiers an vice
versa. Moreover, these quantifiers are universe independent and hence
the logic $\FO_\B(\kI)$ is always regular (see Example \ref{builder}). On
the other hand, without the quantifier $\kI$, these quantifiers can be
more expressive than the corresponding built-in relations \cite{Lu1}.
It is worth noting that the counting extension $\FOC_\B\equiv
\FO_\B(\kI)$ of $\FO_\B$ is not in general the least regular logic
including $\FO_\B$, since, e.g., for $S=\{nk \mid k\in\NN\}$ we have
that
   \[
\FO_{\{\le, S\}} < \FO_{\le}(\kD_n)
< \FO_{\{ \le,S\} }(\kI)\equiv \FOC_{\{\le, S\} },
\]
%
%
where $\kD_n$ is the divisibility quantifier corresponding to the set
$S$. Here the logic $\FO_{\le}(\kD_n)$ is the least regular logic
including $\FO_{\{\le, S\}}$ (see Example \ref{reggap}), and
$\FO_{\le}(\kD_n)< \FO_{\{ \le,S\} }(\kI)$ follows by Theorem~\ref{Idef}.

The definability theory of generalized quantifiers has not yet been
thoroughly developed on ordered structures. The research has concentrated on unary quantifiers: In~\cite{N}, definability of divisibility quantifiers over
ordered structures has been studied. In \cite{Lu1}, a
systematic study of \emph{cardinality quantifiers} $\kC_S$ has been
conducted with an eye on the possibility to define the quantifier
$\kI$ on ordered structures. Cardinality quantifiers are the simplest
kind of unary quantifiers and their definability theory is well
understood over unordered structures. In fact, it is known that on
unordered structures the quantifier $\kI$ cannot be defined in terms
of any cardinality quantifiers~\cite{KV}. On ordered structures, the
situation is very much different. Cardinality quantifiers can be
classified into two cases: the quantifier $\kI$ can be defined in
$\FO_{\le}(\kC_S)$ if and only if $S$ is sufficiently non-periodic
\cite{Lu1} (see Theorem \ref{Idef}). For example, if $S=\{2^n \mid n\in
\NN \}$ or $S=\rg(P)$, where $P$ is a polynomial with nonnegative
integer coefficients of degree at least two, then $\kI$ can be
expressed in $\FO_{\le}(\kC_S)$. In the first part of this paper we
build on this classification of cardinality quantifiers.
In the main result of Section \ref{card}, we formulate a further
combinatorial criterion called \emph{pseudolooseness} for $S\subseteq
\NN$ (see Definition \ref{loose}) which implies that $S$ is
non-periodic in the sense of Theorem \ref{Idef}, and furthermore implies
that $ \FO_{\{+,\times\}}(\Maj) \le \FO_{\le}(\kC_S)$. For instance,
if $S=\rg(P)$ is the range of some polynomial $P$ with nonnegative
integer coefficients of degree at least two, then this criterion is
satisfied, hence
\begin{equation}
\FO_{\le}(\kC_{S}) \equiv \FO_{\{ +,\times\}}(\kI)
\equiv \FO_{\{+,\times\}}(\Maj),  
\end{equation}  
implying that $\FO_{\le}(\kC_S)$ captures  $\TC$. It is worth noting that 
if $S=\rg(P)$ and $P$ is of degree one, 
then all the languages definable in  $\FO_{\le}(\kC_S)$ are regular. 
Also, for  any real number  $r>1$ 
the set 
$S_r=\joukko{\lattia{x^r}\mid x\in \NN  }$ is pseudoloose, and therefore 
   \[
\FO_{\le}(\kC_{S_r}) \ge \FO_{\{+,\times\}}(\Maj).
   \] 
Interestingly, for $E=\{2^n \mid n\in \NN \}$ 
the non-periodicity of $E$ implies that
\[\FO_{\le}(\kC_{E})\equiv \FO_+(\kI,\kC_{E})\]
but since $E$ is not pseudoloose, it remains open whether
$\FO_{\le}(\kC_{E})\equiv \FO_{\{+,\times\}}(\Maj).$

In Section \ref{Builtin} we adapt the familiar properties
of \emph{abstract logics} to accommodate also built-in relations.  
We denote by $\LL_\B$ the logic $\LL$ with built-in relations $\B$.  
We will usually assume that either  $\le$ is definable in $\LL_\B$ 
or $\B=\emptyset$.
We say that $\LL_\B$ is {\em
  semiregular}, if it is closed under $\FO$-operations and has the
\emph{substitution property}.  Also, $\LL_\B$ is {\em regular} if it is
semiregular and closed under relativization (cf.~\cite{E}). 

Without built-in relations, semiregularity
of a logic can be characterized in terms of generalized quantifiers:
$\LL$ is semiregular if and only if there is a class $\Q$ of
quantifiers such that $\LL\equiv\FO(\Q)$.  This characterization
remains valid also for logics with built-in relations once the
 notion of  generalized quantifier is  adapted to the
framework of built-in relations ($\BR$-{\em  quantifiers}):  any 
semiregular logic  $\LL_\B$ with built-in relations $\B$ is  equivalent 
to a logic $\FO_{\B}(\Q)$, where $\Q$ is a class of  $\BR$-quantifiers.
As for regularity, $\FO_\le$ is a regular logic since the restriction
of a linear order to a subset of a model is again a linear order. Note
that, while $\FO_\B$ is usually not regular, $\FO_\le(\Q)$ is regular
for any class $\Q$ of universe independent $\BR$-quanti\-fiers (see
Proposition \ref{foq-reg}). For our purposes, an important consequence
of this observation is that $\FO_{\le}(\kC_S)$ is regular for any
cardinality quantifier $\kC_S$. Inspired by these observations, we
define the notions (adapted from \cite{Lu2}) of \emph{regular interior} 
$\RI(\LL_\B)$ and \emph{regular closure} $\RC(\LL_\B)$ of a
logic $\LL_\B$ with built-in relations. The regular interior of
$\LL_\B$ is the largest regular logic inside $\LL_\B$ and the regular
closure of $\LL_\B$ is the least regular logic including $\LL_\B$.

The Crane Beach Conjecture was a conjecture that arbitrary built-in
relations (besides the order) are of no help in defining languages
with a neutral letter in first-order logic. A symbol $e\in\Sigma$ is a
{\em neutral letter} for a language $L\subseteq\Sigma^*$ if for all
$u,v\in\Sigma^*$ it holds that $uv\in L\iff uev\in L$. In other words,
$e$ is a neutral letter for $L$ if inserting or deleting any number of
$e$'s in a word does not affect its membership in $L$. Since the
property of having a neutral letter is a language theoretic analogue
for the property of being universe independent, it is straightforward
to reformulate the Crane Beach Conjecture as a statement concerning
the regular interior $\RI(\FO_\B)$ of $\FO_\B$. We say that a set $\B$
of built-in relations has the Neutral Letter Collapse Property (NLCP) with
respect to a class $\cal C$ of languages if and only if the
implication
   \begin{equation*}\label{NL}  
L \hbox{ is definable in }\FO_\B\;
\Longrightarrow\; L \hbox{ is definable in }\FO_{\le}
   \end{equation*}
holds for every language $L\in{\cal C}$ with a neutral letter.
Loosely speaking we will then show that $\FO_\B$ has NLCP with respect
to a class $\cal C$ of languages if and only if $\RI(\FO_\B)$
collapses to $\FO_{\le}$ with respect to definability of languages
modulo $\cal C$. For instance, it was shown in \cite{BILST} that $\cal
U$ and $\{+\}$ have NLCP with respect to the class of all languages,
where $\cal U$ contains all unary numerical relations together with
the order $\le$. Hence, it directly follows that over word structures
the logics $\RI(\FO_\B)$ and $\FO_{\le}$ are equivalent, for $\B\in
\{\cal U,\{+\} \}$. We show that this equivalence actually extends to
all vocabularies, that is,
\begin{equation}\label{Rint}
\RI(\FO_\B)\equiv\FO_{\le},
\end{equation}
which implies (and explains) the observation that $\B$ has NLCP with
respect to the class of all languages.

From the computational perspective, the regular closure $\RC(\LL)$ of
a logic $\LL$ is very interesting since most complexity classes are
closed under relativization. Our results imply that if $\B$ contains
the range of a polynomial of degree at least two, 
then $\RC(\FO_\B)\ge \FO_{\{+,\times\}}(\Maj)$ in contrast to \eqref{Rint}. 
In particular, the regular
closure of $\AC$ is $\TC$.

\paragraph{This article is organized as follows:} In Section 2 we
review concepts and previous work on generalized quantifiers relevant
for the results of Section 3. In Section 3 we introduce and analyze
the notion of \emph{(pseudo)looseness} for subsets of natural numbers
and show how to define multiplication in terms of a pseudoloose
cardinality quantifier. In Section 4 we turn to built-in relations and
adapt the notions of generalized quantifiers and regularity of logics
from Abstract Model Theory to the setting with built-in relations. In
Section 5 we define the notions of regular interior and closure for a
logic with built-in relations and relate the former to the Crane Beach
Conjecture.  Section~6 also contains results demonstrating the gap
between the regular interior and closure of first-order logic with
built-in relations.

\section{Background and preliminaries}\label{Background}

In this section we recall concepts and previous results which are
addressed in Section~\ref{card}. In particular, a more detailed
exposition of generalized quantifiers (with built-in relations) is
provided in Section~\ref{Builtin}.

\subsection*{Some notation}

The set of natural numbers is denoted by $\NN=\{0,1,2,\ldots\}$. The
set of integers is denoted by $\ZZ$, and $\ZZ_+=\{1,2,3,\ldots\}$
denotes the set of positive integers. For a set $X$, the power set of
$X$ is denoted by $ \mathcal{P}(X)=\{ Y\ |\ Y\subseteq X \}$, and
for $E\subseteq X^2$, the domain $\dom(E)$ of the relation $E$ is the set
\[\dom(E)=\{a\in X\ |\ (a,b)\in E,\textrm{ for some $b\in X$} \},\]
 and the range $\rg(E)$ of  $E$ is  
\[\rg(E)=\{b\in X\ |\ (a,b)\in E,\textrm{ for some $a\in X$}\}.\]

 For logics $\LL$ and $\LL^*$, the logic  $\LL$ is  at most as strong
as $\LL^*$, in symbols
$\LL\le\LL^*$, if every class $K\subseteq\Str(\tau)$ which
is definable in $\LL$ is also definable in
$\LL^*$.  If   $\LL\le\LL^*$ and $\LL^*\le\LL$, we write $\LL\equiv\LL^*$,
and say that $\LL$ and $\LL^*$ are equivalent.
For logics on ordered structures, these notions are defined
analogously (see Section \ref{Builtin} for more details). Finally, for a logic $\mathcal{L}$ and a complexity class
$C$, we write $\LL\equiv C$ if for all $\Sigma$, and
$L\subseteq \Sigma^+$: $L\in C$ if and only if the class of word
models (see Example \ref{bqex}) corresponding to $L$ is definable in
$\LL$.
   
\subsection*{Generalized quantifiers}
In this subsection we review 
the generalized quantifiers discussed in
the introduction. However, we postpone
the formal definition of generalized quantifiers to
Section~\ref{Builtin}. The notion of a generalized quantifier goes
back to \cite{M} and \cite{Lin}. For a more complete account on
quantifiers, see \cite{KV}.
 
Examples of unary quantifiers of vocabulary $\{U\}$ are the 
{\it divisibility quantifier} $\kD_n$ expressing that the size of a unary
relation is divisible by~$n$ ($n\in\NN$) and, more generally, for a
fixed $S\subseteq\NN$ the {\it cardinality quantifier $\kC_S$}.  For
every $\mM\in\Str(\{U\})$,
   \[
\mM\models \kC_Sx\,(U(x)) \quad\text{if and only if}\quad |U^{\mM}|\in S.
   \]
Note that $\kD_n=\kC_S$ with $S=n\NN=\{nk\ |\ k\in\NN\}$.
In a quantifier logic, such as  $\FO(\kC_S)$, the quantifiers can be nested in the natural way, so that in the rule above, $U$ can be
replaced by a formula~$\psi$.
For example, if $\mA$ is any structure and $\psi(x)$ is a formula
in the appropriate language, then 
   \[
\mA\models \kD_2x\,(\psi(x)) \quad\text{if and only if}
\quad |\psi^{\mA}|\text{ is finite and of even size},
   \]
where $\psi^{\mA}=\joukko{a\in\Dom(\mA) \mid \mA\models\psi[a/x]}$,
of course.

The quantifier $\Maj$ also has  vocabulary~$\{U\}$: 
   \[
\mM\models \Maj\, x\,(U(x)) \quad\text{if and only if}\quad 
|U^{\mM}|>\card(\mM)/2.
   \]
For the purposes of this paper, \emph{universe independence} (see
Section \ref{Builtin}) is an important property of the quantifiers~$\kC_S$.
Universe independence, in the case of the quantifiers~$\kC_S$,
means simply that the truth of the formula $\kC_Sx\,(U(x))$
depends only on $|U^{\mM}|$, disregarding $|\Dom( \mM )\setminus
U^{\mM} |$.  Note that the quantifier $\Maj$ is an example of a
quantifier which is not universe independent.

The {\em equicardinality} or {\em the H\"artig quantifier}~$\kI$,  
and the {\it general divisibility quantifier}~$\kD$ are quantifiers of
vocabulary $\{U,V\}$.  For $\mM\in\Str(\{U,V\})$, we have
   \[ 
\mM\models \kI\, x,y\,(U(x),V(y)) 
\quad\text{if and only if}\quad 
|U^{\mM}|=|V^{\mM}|,
   \]
and
   \[
\mM\models \kD\, x,y\,(U(x),V(y))
\quad\text{if and only if}\quad 
|U^{\mM}|\,\Big\vert\,|V^{\mM}|.
   \]

H\"artig showed in \cite{Ha} that addition, as a ternary predicate
$+$, can be defined in terms of the quantifier $\kI$ and first-order
logic on ordered structures. Indeed, it is easy to verify that $+$ is
defined by the formula $\phi(x,y,z)$, where
  \begin{equation}\label{Hartig's trick}
\phi(x,y,z)\sij\kI u,v(u<x,y<v\le z).
   \end{equation}
It was observed in \cite{Lu1} that
\begin{equation}\label{MajI} 
\FO_{\le}(\Maj)\equiv\FO_{\le}(\kI)\equiv \FOC_{\le},  
\end{equation}
where $\FOC$ is the extension of $\FO$ in terms of counting
quantifiers $\exists ^{=y} x$ \cite{IL}.  The counting quantifiers  work syntactically so that, in the formula $\exists ^{=y} x\,\psi(x,\vz) $,  variable $x$ is bound and $y$ is  a new free
variable. Let  $\mM$ be an ordered structure with  
$\Dom(\mM)=\{0,\ldots,n-1\}$ and
$\le^\mM$ the natural order. Then the semantics of $\exists ^{=y} x$
is given by
   \[
\mM\models\exists ^{=y} x\,\psi(x,\vz)[c/y,\vb/\vz ]
\quad\text{if and only if}\quad
|\{ a\in \Dom(\mM)\mid  \mM\models \psi[a,\vb/\vz]\}|= c.
\]

It is interesting to note that, by \eqref{MajI}, $\FOC_{\le}$ can be
replaced by either of the quantifier logics $\FO_{\le}(\Maj)$ or
$\FO_{\le}(\kI)$, which have well-defined counterparts, $\FO(\Maj)$
and $\FO(\kI)$, on unordered structures. Note that $\FOC$ does not
make sense without order.  Further note that, of the quantifiers
$\Maj$ and $\kI$, the latter is also universe independent.

 It was shown in \cite{Li}, \cite{R} and \cite{LMcKSV}, that
multiplication cannot be defined in the logic $\FO_{\le}(\Maj)$. The
following logics are therefore strictly more expressive than the
logics in \eqref{MajI}:
\begin{equation}\label{logicsforTC}
 \FO_{\{\le,\times\}}(\kI)\equiv \FO_{\{+,\times\}}(\Maj)\equiv  \FO_\bit(\Maj)
\equiv \FO_{\le}(\Maj^2)\equiv  \FO_{\le}(\kD).
\end{equation}
The first equivalence follows by the definability of addition in $\FO_{\le}(\kI)$ (see \eqref{Hartig's trick}), and~\eqref{MajI}. 
The second equivalence follows from the fact that $\FO_{\{+,\times\}}\equiv \FO_\bit$ \cite{I}. The third equivalence, where $\Maj^2$ denotes 
the second vectorization of $\Maj$  (i.e., majority of pairs),  
was shown in \cite{BIS}, and the last equivalence is due to \cite{Lu1}. 
It is worth noting that, unlike $\Maj^2$, $\kD$~is a unary quantifier. 
Finally, we note that the logics in~\eqref{logicsforTC} are regular logics. 
This follows by Proposition \ref{foq-reg} applied to $\FO_{\le}(\kD)$.    

As a demonstration of the power of regularity, we show that the
undefinability of multiplication in any of the logics in \eqref{MajI}
is actually a direct corollary of an old result of Krynicki and
Lachlan.

\begin{theorem}[\cite{KL}]
The $\FO(\kI)$-theory of $\mN=\langle \NN,{+} \rangle$, i.e.,
the set of $\FO(\kI)$-sentences true in~$\mN$, is decidable.   
\end{theorem}

\begin{corollary}
Multiplication is not $\FO_+(\kI)$-definable.  
\end{corollary}

\begin{proof}
Suppose towards contradiction that there were a formula~$\mu(x,y,z)$
of $\FO(\kI)$ defining multiplication in finite structures with
built-in addition. Working now in 
the structure~$\mN=\langle \NN,{+}\rangle$, 
we observe that we can interpret the structure~$\langle
n,{+}\rangle$ given one parameter~$n\in\ZZ_+$.  By regularity of
$\FO(\kI)$, there is an $\FO(\kI)$-formula~${\tilde\mu}(x,y,z,p)$ of the
vocabulary~$\joukko{+}$ such that
   \[
\mN\models{\tilde\mu}[a/x,b/y,c/z,n/p]
\text{ if and only if }
\langle n,{+}\rangle\models\mu[a/x,b/y,c/z]
\text{ and }a,b,c<n,
   \]   
for $a,b,c\in\NN$.  The latter condition is equivalent to
$ab=c<n$ and $a,b<n$.  Hence, 
$\nu(x,y,z)\colon \exists p\,{\tilde\mu}(x,y,z,p)$
is a $\FO(\kI)$-formula that defines multiplication on $\mN$.

Given a sentence $\phi\in\FO[\joukko{{+},{\times}}]$,
there is an effective way to do the substitution which gives
$\phi(\nu/\times)\in\FO(\kI)[\joukko{+}]$.   Now we have
   \[
\langle \NN,{+},{\times}\rangle\models\phi
\quad\text{if and only if}\quad
\mN\models\phi(\nu/\times),
   \]
which reduces the $\FO$-theory of $\langle \NN,{+},{\times}\rangle$
to the decidable $\FO(\kI)$-theory of $\mN$.
However, this is in contradiction with one of the most fundamental
results of mathematical logic that the $\FO$-theory of 
$\langle \NN,{+},{\times}\rangle$ is undecidable.
\end{proof}

\subsection*{Cardinality quantifiers expressing equicardinality}

In this section we briefly recall the result \cite{Lu1}
characterizing cardinality quantifiers~$\kC_S$ which can define the
quantifier $\kI$ on ordered structures.
The periodicity of the set $S$ is measured in terms of the following
functions $f_S$ and $\oo_S$.

\begin{enumerate}[(a)]

\item
Let $f\colon \NN\to\NN$ and $\Delta\osaj\NN$ be an interval.
Then $f$ is {\it periodic on $\Delta$ with period} $\oo\in\NN$
if $f(x)=f(x+\oo)$ whenever $x,x+\oo\in\Delta$.  A set 
$S\osaj\NN$ is {\it periodic on $\Delta$ with period $\oo$}
if its characteristic function $\kh_S$ is, i.e., 
$x,x+\oo\in\Delta$ implies $x\in S\iff x+\oo\in S$.

\item
Let $S\osaj\NN$.  The functions $f_S,\oo_S\colon \NN\to\NN$
are defined in the following way.  Let $n\in\NN$.  Then
$f_S(n)$  is the least $\ell\in\NN$ such that for some
$\oo\in\NN$, $0<\oo\le\ell$, the set $S$ is periodic
on the interval $\{i\in\NN\ |\ \ell-\oo\le i\le n-(\ell-\oo)\}$
with period~$\oo$.  Furthermore, $\oo_S(n)$ is the least
$\oo\in\ZZ_+$ 
such that $S$ is periodic on the interval
$\{i\in\NN\ |\ f_S(n)-\oo\le i\le n-(f_S(n)-\oo)\}$
with period $\oo$.
\end{enumerate}

\begin{theorem}[\cite{Lu1}]\label{Idef}
Let $S\osaj\NN$.  Then the following are equivalent:

\begin{enumerate}

\item \label{NPC} There are $k,\ell\in\NN$
such that $k\cdot f_S(n)\cdot\oo_S(n)^\ell\ge n$,
for almost all $n\in\NN$.  
\item  $\FO_{\le}(\kI)\le\FO_{\le}(\kC_S)$.  
\end{enumerate}
\end{theorem}
In \cite{Lu1}, the inequality in condition \ref{NPC} of Theorem
\ref{Idef} was required to hold for every $n\in \NN$. It is easy to see
that, if a set $S$ satisfies condition \ref{NPC} with $k,\ell\in\NN$,
then, by replacing $k$ by a big enough $k'$, the inequality will be
satisfied by every $n\in \NN$.

Next we apply Theorem \ref{Idef} to certain interesting sets $S\osaj\NN$.

\begin{lemma}[\cite{Lu1}]\label{Idefex}
\begin{enumerate}[(a)]

\item Let $S=\rg(P)$ be the range of some polynomial $P$ with
  nonnegative integer coefficients of degree at least two. Then
  the quantifier $\kI$ is definable in $\FO_{\le}(\kC_S)$.

\item \label{Idef-E} Let $E=\{2^n \mid n\in\NN\}$. Then
  the quantifier $\kI$ is definable in $\FO_{\le}(\kC_E)$.

\item Let $F=\{n!\mid n\in\NN\}$. Then
  $\FO_{\le}(\kI)\not\le\FO_{\le}(\kC_F)$.

\end{enumerate}
\end{lemma}

\begin{proof}
In view of Theorem~\ref{Idef}, we need to show that the non-periodicity
condition~\ref{NPC} holds in the first two cases and fails in the
last case.  What is common to all of these cases is that the
parameter set~$U\osaj\NN$ of the cardinality quantifier $\kC_U$ considered
is a range of a strictly increasing sequence of natural numbers
with strictly increasing differences of consecutive elements,
i.e., there is an enumeration
$U=\{ u_k \mid k\in\NN \}$ such that $0<u_{k+1}-u_k<u_{k+2}-u_{k+1}$,
for each $k\in\NN$.

The following fact on such sets $U$ is useful:  Let $\Delta\osaj\NN$
be an interval and $\oo\in\ZZ_+$. Suppose $U$ is periodic on  $\Delta$
with period $\oo$.  Then if $|\Delta|\ge 2\oo$, then
$|\Delta\cap U|\le 2$.  Indeed, suppose to the contrary that
$|\Delta|\ge 2\oo$ and $|\Delta\cap U|\ge 3$.  Let $a$, $b$ and $c$
be the three least elements of $\Delta\cap U$ with $a<b<c$.
Obviously $a-\oo\not\in\Delta$, so $|\Delta|\ge 2\oo$
implies $a+\oo\in\Delta$.   Since $U$ is periodic on $\Delta$
with period $\oo$, we have that $a+\oo\in U$.  As $b$ is the
second least element of $\Delta\cap U$, it holds that $b\le a+\oo$.
If $b=a+\oo$ were true, we would have $b+\oo<c$, as $c-b>b-a=\oo$
(the differences are increasing),
so $b+\oo\in\Delta$ and, by periodicity, $b+\oo\in U$.
However, this is impossible, as $c$ is the third least
element of $\Delta\cap U$.  Hence, $a+\oo>b$.
But then again we notice that $b-\oo\not\in\Delta$ and
$b+\oo\in\Delta\cap U$, so $a+\oo,b+\oo\in\Delta\cap U$
which is in contradiction with the assumption that the differences
of consecutive elements of $U$ strictly increase.
So under our assumptions, $|\Delta\cap U|\le 2$.

(a)  
As $\deg(P)\ge 2$ and the coefficents are nonnegative, our
observation holds.  Clearly $\lim_{k\to\infty}P(k+1)/P(k)=1$.
This implies that for almost all $k\in\NN$, we have
that 
$$
  1<P(k+1)/P(k)<\root3\of{3},
$$
which in turn implies that for almost all $n\in\NN$,
the interval 
$$
  \Delta=\{ \lattia{n/4}, \lattia{n/4}+1,\ldots, \katto{3n/4}-1,\katto{3n/4} \}
$$
contains at least three elements of $U$.  We claim that $f_S(n)\ge n/4$.
This is because $S$ is periodic on the interval
$$
   \Delta_0=\{ f_S(n)-\oo_S(n),\ldots,n-(f_S(n)-\oo_S(n))\}
$$
with period $\oo_S(n)$.  If $f_S(n)<n/4$ were true, then
obviously $\Delta\osaj\Delta_0$, so $S$ would also be periodic on $\Delta$
with period $\oo_S(n)$.  In addition, we would have
$|\Delta|\ge n/2\ge 2f_S(n)\ge 2\oo_S(n)$, but this all is
in contradiction with our observation.  So $f_S(n)\ge n/4$,
which easily implies the non-periodicity condition~\ref{NPC}.

(b)  This case is similar to the previous case, only that the
details are easier.  For every $n\in\NN$ with $n\ge 9$,
the interval
$$
  \Delta=\{ \lattia{n/9}, \lattia{n/9}+1,\ldots, \katto{8n/9}-1,\katto{8n/9} \} 
$$
contains at least 3~elements of~$E$, as the ratio of consecutive
elements of~$E$ is $2=\root3\of 8$.  This implies $f_E(n)\ge n/9$,
so condition~\ref{NPC} is fulfilled.

(c)  Let $n\in\NN$, $n\ge 3$.  Consider the value of the functions
$f_F$ and $\oo_F$ at $m=(n+1)!-1$.  Put $\Delta=\{ n!+1,\ldots, m-n!-1 \}$,
whence $\Delta\cap F=\emptyset$.  This implies that
$F$ is trivially periodic on $\Delta$ with period~$1$, so
$f_F(m)\le n!+2$.  However, if $\ell,\oo\in\NN$ are such that
$0<\oo\le \ell\le n!+1$ and
$\Gamma=\{ \ell-\oo, \ell-\oo+1,\ldots,m-\ell-\oo\}$,
then $n!,n!+\oo\in\Gamma$, but $n!\in F$ 
and $n!+\oo\not\in F$.
Hence, $f_F(m)=n!+2$ and $\oo_F(m)=1$.
This means that $f_F(m)\oo_F(m)^s=n!+2$ irrespective of the exponent~$s$.
We observe that
$$
   \lim_{n\to\infty}\frac{n!+2}{(n+1)!-1}=0,
$$
which implies that the non-periodicity condition fails.
\end{proof}

\subsection*{Complexity classes}

Next we recall the complexity classes relevant for this article. In
this article $\AC$ and $\TC$ refer to the classes of languages
recognized by $\DLOGTIME$-uniform families $(C_n)_{n\in\NN}$ of
constant depth polynomial-size circuits. For $\AC$, the circuit $C_n$
may have NOT gates and unbounded fan-in AND and OR gates. For $\TC$, also
unbounded fan-in MAJORITY gates are allowed, which output 1 if and only if at
least half of the inputs are 1. The requirement of
$\DLOGTIME$-uniformity means that $(C_n)_{n\in\NN}$, as a family of
directed acyclic graphs, can be recognized by a random access machine
in time $O(\log(n))$ (see \cite{V} for details).

In this article we are concerned with the logical counterparts of
$\AC$ and $\TC$:
\begin{align*}
 \AC &\equiv \FO_{\{+,\times\}}\\
 \TC &\equiv \FO_{\{+,\times\}}(\Maj) \equiv \FO_{\{+,\times\}}(\kI)\equiv  \FO_{\{\le,\times\}}(\kI).
\end{align*}
These logical characterizations were proved in the seminal paper
\cite{BIS}, where the connection between $\DLOGTIME$-uniformity
and $\FO_{\{+,\times\}}$-definability was established.

\section{Cardinality quantifiers and $\TC$ }\label{card}

In this section we obtain new characterizations of  $\TC$ in
terms of certain cardinality quantifiers.

Let $\Sq=\{n^2\ |\ n \in \NN \}$ be the set of squares and let us
consider the quantifier $\kC_{\Sq}$.  By Lemma \ref{Idefex}, 
the H\"artig quantifier $\kI$ is definable in $\FO_{\le}(\kC_{\Sq})$.
On the other
hand, Lynch showed in \cite{Ly} that multiplication is already
definable in terms of addition and the relation $\Sq$. From this it
directly follows that

\begin{equation}\label{sqchar}
\FO_{\le}(\kC_{\Sq}) \equiv \TC,  
\end{equation}
since the quantifier $\kC_{\Sq}$ is easily seen to be expressible in
$\FO_{\{+,\times\}}(\kI)$. In this section we show that, in
\eqref{sqchar}, $\Sq$ can be replaced by numerous sets~$S$ of
natural numbers, among others by the range $\rg(P)$ of any
polynomial $P$ with nonnegative integer coefficients of degree at
least two.  

We use only elementary methods to achieve our goal.
However, we need several steps to complete the proof.
The key point is the elementary combinatorial fact that if
$(A_i)_{i\in I}$ is a finite disjoint family of finite sets of equal
size, then
   \[
\bigl|\bigcup_{i\in I}A_i\bigr|=|I|\;|A_{i_0}|
   \]
with $i_0\in I$ arbitrary.  Actually, this fact can be construed as
a combinatorial definition of multiplication on natural numbers.
We shall use cardinality quantifiers~$\kC_S$ for a partial
logical implementation of this definition.

A \emph{partial multiplication} is a ternary relation~$R$ on natural numbers
such that $(a,b,c)\in R$ implies $a\cdot b=c$.
We are going to show that under certain circumstances, it is
possible to extend a partial multiplication to the multiplication
restricted to the universe at hand.
We shall proceed in three steps, each corresponding
to a subsection:

\begin{enumerate}
\item
We examine in the subsection 
\emph{Extending partial multiplication}  how partial multiplication
can be extended in first order logic.  There are, admittedly,
easier ways to do this than the one which we employ, but
the method here is tailor-made for the application
related to cardinality quantifier logics.  To be more precise,
it would be possible to control the extension of partial multiplication
by some single invariant, but we use a control function,
which will be denoted by $\g(R,k)$, instead. We get sufficient bounds
for this control function so that a partial multiplication
can be extended to multiplication restricted to the whole domain
of the finite structure.

\item
In the subsection \emph{Pseudoloose sets},
we extract from the previous analysis
the notions of looseness and pseudolooseness of a set of natural numbers,
which give sufficient conditions for the cardinality quantifier
so that multiplication restricted to the structure is definable
in the related quantifier logic in ordered structures.
In consequence, we have that
these cardinality logics are at least as strong as the complexity
class $\TC$ in ordered structures.

\item
In the third subsection \emph{Analyzing pseudolooseness}, we work out
some concrete examples of cardinality quantifier logics where
defining multiplication is possible.
\end{enumerate}

We emphasize that pseudolooseness gives only a sufficient condition for
the cardinality quantifier logics in order to the restricted multiplication
be definable in ordered structures, in contrast to the results of
\cite{Lu1} where a corresponding sufficient and necessary condition
is given for the restricted addition to be definable.
We simply anticipate that a complete characterization might be
too involved to be truly feasible.

\subsection*{Extending partial multiplication}

For $k\in\NN$ and a partial multiplication $R$,  we define
$\g(R,k)$ to be the biggest $r\in\NN$ such that $r=0$ or
for every $a,b\in\NN$ with $a\le k$ and $b\le r$, we have
$(a,b,ab)\in R$. We fix ternary relation symbols $A$ and $M$. 
By a \emph{partial model of arithmetic}
we understand a finite $\{A,M\}$-structure~$\mN$ 
such that $\Dom(\mN)=\{0,1,\ldots,n-1\}$,
$A^\mN$ is addition restricted to~$n$, i.e., 
$A^\mN=\{(a,b,c)\in\Dom(\mN)^3\mid a+b=c\}$
and $M^\mN$ is a partial multiplication.

The mapping $\g$ has the following obvious monotonicity property:
If $R\osaj R'\osaj\NN^3$ and $k,k'\in\NN$, $k\ge k'$,
then $\g(R,k)\le\g(R',k')$.  In addition, we observe that
if $\mN$ is a partial model of arithmetic with $n=\card(\mN)$
and $0<k<n$, then $\g(M^\mN,k)\le\idiv{n-1}{k}$.

\begin{lemma} \label{mulextstep}
There exists a first-order formula $\mu(x,y,z)$ such that for
any partial model of arithmetic~$\mN$, we have that $\mu^\mN$
is also a partial multiplication and 
for all $a,b\in\Dom(\mN)\smallsetminus\{0\}$,

\begin{enumerate}[(a)]
\item
$\g(\mu^\mN,a)\ge \g(M^\mN,a)$,

\item
$\g(\mu^\mN,a)\ge b$ iff $\g(\mu^\mN,b)\ge a$ and

\item
if $a<b\le a^2+a$, then
   \[
\g(\mu^\mN,b)\ge
  \min\left\{\lattia{\frac{\g(M^\mN,a)}{\lattia{(b-1)/a}}},
              \idiv{n-1}{b}\right\}
   \]
where $n=\card(\mN)$.
  
\end{enumerate}

\end{lemma}

\begin{proof}
We use the basic algebraic laws of multiplication for
extending~$M^\mN$:  The idea is 
that given $x,y\in\Dom(\mN)$,
if we are able to represent $x$ as $x=tu+t'u'$, then
we can calculate $xy=t(uy)+t'(u'y)=tv+t'v'$ with
$v=uy$ and $v'=u'y$ provided that the appropriate products
are definable by $M^\mN$.  More formally, consider the formulas
\begin{align*}
\alpha(x,t,u,t',u')&\sij 
    \exists s\exists s'\, (M(t,u,s)\land M(t',u',s') \land A(s,s',x)), \\
\mu_-(x,y,z) &\sij
    \exists t\exists t' \exists u\exists u' \exists v\exists v'\,
       (M(u,y,v)\land M(u',y,v') \\ 
    &\phantom{\sij\exists t\exists t' \exists u\exists u' 
                                      \exists v\exists v'(}
      \land \alpha(x,t,u,t',u') \land \alpha(z,t,v,t',v'))\\
\noalign{and} 
\mu(x,y,z)& \sij\mu_-(x,y,z)\lor \mu_-(y,x,z). 
\end{align*}
Let $\mN$ be a partial model of arithmetic.  
For $a,d,d',e,e'\in\Dom(\mN)$, it holds that 
   \[
\mN\models\alpha[a/x,d/t,d'/t',e/u,e'/u']\text{ implies }a=de+d'e'.
   \]
Furthermore, for $a,b,c\in\Dom(\mN)$ we have that
$\mN\models\mu_-[a/x,b/y,c/z]$ iff there are $d,d',e,e'\in\Dom(\mN)$
such that $a=de+d'e'$, $c=deb+d'e'b=ab$ and
   \[
\joukko{(d,e,de),(d',e',d'e'),(e,b,eb),(e',b,e'b),
(de,b,deb),(d'e',b,d'e'b)}\osaj M^\mN.
   \]
Hence, $\mu^\mN_-$ is a partial multiplication, and by commutativity
of multiplication, this holds for $\mu^\mN$, too.

The lower bounds are now easy to derive.

\begin{enumerate}[(a)]

\item
Clearly, $\mu^\mN_-\osaj\mu^\mN$ and we may use representations of
the form $a_0=a_0\cdot 1+0\cdot 0$ to show that 
$\g(M^\mN,a)\le\g(\mu^\mN_-,a)$.  Hence,
$\g(\mu^\mN,a)\ge\g(\mu^\mN_-,a)\ge \g(M^\mN,a).$

\item
By symmetry of $\mu$ and the definition of $\g$, this is obvious.

\item
Suppose $a,b\in\Dom(\mN)$ satisfy $0<a<b\le a^2+a$.
Put 
   \[
d=\min\joukko{\lattia{\frac{\g(M^\mN,a)}{\lattia{(b-1)/a}}},
              \idiv{n-1}{b}}.
   \]
Then $d\le\g(M^\mN,a)<n$, $bd<n$, and we are to show that
$\g(\mu^\mN,b)\ge d$.  We may assume that $d>0$, which implies
$\g(M^\mN,a)>0$.

Let $b_0,d_0\in\NN$ where $b_0\le b$ and $d_0\le d$.
Consider the representation $b_0=a\cdot q+r\cdot 1$
where $q=\idiv{b_0-1}{a}\le\idiv{a^2+a-1}{a}=a$
and $r\le a$.  Since $d_0\le d\le\g(M^\mN,a)$, we immediately
see that $(q,d_0,qd_0),(1,d_0,d_0),(r,1,r),(r,d_0,rd_0)\in M^\mN$.
The critical case is that of the triple $(a,qd_0,aqd_0)$,
which is in $M^\mN$, because 
   \[
qd_0\le\idiv{b-1}{a}\cdot d
\le\idiv{b-1}{a}\cdot\lattia{\frac{\g(M^\mN,a)}{\lattia{(b-1)/a}}}
\le\g(M^\mN,a).
   \]  
Clearly also $(a,q,aq)\in M^\mN$.
As in all desired 
cases we may multiply using $M^\mN$,
we get $(b_0,d_0,b_0d_0)\in\mu^\mN_-\osaj\mu^\mN$.
Consequently, $\g(\mu^\mN,b)\ge d$. \qedhere
\end{enumerate}

\end{proof}

\begin{proposition} \label{mulext}
For $k\in\NN$ with $k\ge 2$, there exists an $\{A,M\}$-formula $\pi(x,y,z)$
of~$\FO[\{A,M\}]$ such that the following condition holds:
Let $\mN$ be a  partial model of arithmetic with $\card(\mN)=n\ge k^2$.
Suppose that there is $a^*\in\Dom(\mN)$ such that
   \[
n^{1/k}\le a^*\le n^{1-1/k}/k
\quad\text{ and }\quad
ka^*\g(M^\mN,a^*)\ge n.
   \]
Then $\pi$ is the multiplication restricted to~$n=\Dom(\mN)$.
\end{proposition}

\begin{proof}
We substitute $\mu$ of the preceding lemma repeatedly for $M$,
getting $\pi_0(x,y,z)=M(x,y,z)$  and $\pi_{i+1}=\mu(\pi_i/M)$, for $i\in\NN$.
We claim that $\pi_{i*}$ with $i^*=2\katto{\lb k}+2$ has the
required property, where $\lb k$ denotes the binary logarithm of $k$.  Suppose $\mN$ satisfies the assumptions of the lemma.
Consider the partial model of arithmetic $\mN_i$ with $\Dom(\mN_i)=n$
and $M^{\mN_i}=\pi^\mN_i$. Then for every $i\in\NN$ and 
$a,b\in\joukko{1,\ldots,n-1}$ the preceding lemma gives
\begin{enumerate}[(a)]

\item \label{gammainc}
$\g(\pi^\mN_{i+1},a)\ge \g(\pi^\mN_i,a)$,

\item \label{gammasymm}
$\g(\pi^\mN_{i+1},a)\ge b$ iff $\g(\pi^\mN_{i+1},b)\ge a$ and

\item \label{gammaext}
if $a<b\le a^2+a$, then
   \[
\g(\pi^\mN_{i+1},b)\ge
  \min\joukko{\lattia{\frac{\g(\pi^\mN_i,a)}{\lattia{(b-1)/a}}},
              \idiv{n-1}{b}}.
   \]

\listpart{Furthermore, these inequalities imply for 
$a\in\joukko{1,\ldots,n-1}$ and $i\in\ZZ_+$ that} 

\item \label{gammadbl}
$\g(\pi^\mN_{i+1},a)\ge\min\joukko{2\g(\pi^\mN_i,a),\idiv{n-1}{a}}$.
  
\end{enumerate}

In order to prove this last statement, put $b=\g(\pi^\mN_i,a)$
and $c=\min\joukko{2b,\idiv{n-1}{a}}$ so that
we are to prove $\g(\pi^\mN_{i+1},a)\ge c$. If $b=0$ 
or $b=\idiv{n-1}{a}$, then $b=c$ and the inequality
follows directly from inequality~\ref{gammainc}.  So suppose
that $0<b<\idiv{n-1}{a}$.  Then by symmetry (item~\ref{gammasymm}),
we have $\g(\pi^\mN_i,b)\ge a$, and as $0<b<c\le b^2+b$,
we have by estimate~\ref{gammaext} that
\begin{align*}
\g(\pi^\mN_{i+1},c) 
   & \ge\min\joukko{\lattia{\frac{\g(\pi^\mN_i,b)}{\lattia{(c-1)/b}}},
              \idiv{n-1}{c}} \\
   & \ge\min\joukko{\g(\pi^\mN_i,b),
                    \lattia{\frac{n-1}{\idiv{n-1}{a}}}}\ge a. \\
\end{align*}
Again by symmetry, $\g(\pi^\mN_{i+1},a)\ge c$.

Set $a_1=\min\joukko{a^*,\g(M^\mN,a^*)}$ and recursively
$a_{i+1}=\min\joukko{a^2_i,n-1}$, for $i\in\ZZ_+$.  Then we have that
\begin{shortenumerate}
\item $a^2_1<n$,
\item $a^k_1\ge n$,
\item $ka_1\g(\pi^\mN_i,a_1)\ge n$ \quad and
\item $\g(\pi^\mN_i,a_1)\ge a_1$,  
\end{shortenumerate}
for $i\in\ZZ_+$.
The first inequality follows from 
$a^2_1\le a^*\g(M^\mN,a^*)\le a^*\idiv{n-1}{a^*}<n$.
If $a_1=a^*$, the rest of the inequalities are clear, so assume
$a_1=\g(M^\mN,a^*)$.  Then by the assumptions of the lemma, we have
$ka^*a_1\ge n$ and $ka^*\le n^{1-\frac1k}$, so $a_1\ge n^{1/k}$,
implying the second inequality.  By \ref{gammainc} and \ref{gammasymm},
we have $\g(\pi^\mN_i,a_1)\ge\g(M^\mN,a_1)\ge a^*\ge a_1$, i.e.,
last inequality.  Consequently, $ka_1\g(\pi^\mN_i,a_1)\ge ka_1a^*\ge n$.

Consider $f\colon \joukko{a_1,\ldots,n-1}\to n$, 
   \[
f(x)=\min\joukko{\lattia{\frac{\gamma(\mu^\mN,a_1)}{\lattia{x/a_1}}},
                  \idiv{n-1}{x}}.
   \]
By induction on $i\in\ZZ_+$, we get 
for $b\in\joukko{a_1,a_1+1,\ldots,a_i}$ that $\g(\pi^\mN_i,b)\ge f(b)$.
Indeed, the case $i=1$ is trivial. Supposing the induction hypothesis
holds for~$i$, we need only to prove the inequality 
$\g(\pi^\mN_{i+1},b)\ge f(b)$ for $b\in\joukko{a_i+1,\ldots,a^2_i}$
as $\g(\pi^\mN_{i+1},b)\ge\g(\pi^\mN_i,b)$. Then $a_i<b\le a^2_i+a_i$
so by induction hypothesis and inequality~\ref{gammaext}, we have
\begin{align*}
\g(\pi^\mN_{i+1},b) 
     & \ge\min\joukko{\lattia{\frac{\g(\pi^\mN_i,a_i)}{\lattia{(b-1)/a_i}}},
              \idiv{n-1}{b}}   \\
     & \ge\min\joukko{\lattia{\frac{f(a_i)}{\lattia{b/a_i}}},
              \idiv{n-1}{b}}
       =\lattia{\min\joukko{\frac{f(a_i)}{\lattia{b/a_i}},
               \frac{n-1}{b}}}   \\
     & =\lattia{\min\joukko{\frac{\g(\mu^\mN,a_1)}
                               {\lattia{b/a_i}\lattia{a_i/a_1}},
                \frac{n-1}{b}}} \\
     & \ge f(b),
\end{align*}
as for arbitrary $r,s\ge 0$, it holds
that $\lattia{rs}\ge\lattia{r}\lattia{s}$.

Obviously $a_i=\min\joukko{a^{2^{i-1}}_1,n-1}$, for $i\in\ZZ_+$.
In particular, we get $a_{i^*/2}=a_{\katto{\lb k}+1}=n-1$.
Hence, for $b\in\joukko{a_1,a_1+1,\ldots,n-1}$, it holds that
   \[
\g(\pi^\mN_{i^*/2},b)\ge f(b)
\ge\min\joukko{\lattia{\frac{n/(ka_1)}{\lattia{b/a_1}}},\idiv{n-1}{b}}
\ge\idiv{n}{kb},
   \]
as $ka_1\g(\mu^\mN,a_1)\ge n$.  We would like to have the same estimate
for $b\in\joukko{1,\ldots,a_1-1}$, too. If $\idiv{n}{kb}<a_1$,
then by natural monotonicity of $\g$, we have
   \[
\g(\pi^\mN_{i^*/2},b) \ge \g(\pi^\mN_{i^*/2},a_1)
\ge a_1 > \idiv{n}{kb},
   \]
so suppose $\idiv{n}{kb}\ge a_1$.  Then we may apply the estimate
for $\idiv{n}{kb}$, getting
   \[
\g(\pi^\mN_{i^*/2},\idiv{n}{kb})
\ge \lattia{\frac{n}{k\idiv{n}{kb}}}
\ge b, 
   \]
and by symmetry (item \ref{gammasymm}),
$\g(\pi^\mN_{i^*/2},b)\ge\idiv{n}{kb}$.

To finish the proof, we show that for every $b\in\joukko{1,\ldots,n-1}$,
we have $\g(\pi^\mN,b)=\g(\pi_{i^*},b)=\idiv{n-1}{b}$,
which is equivalent to claiming that $\pi^\mN$ is the multiplication
restricted to $n$.

Applying inequality~\ref{gammadbl} repeatedly, we get
\begin{align*}
\g(\pi^\mN_{i^*/2+j},b)
    &\ge\min\joukko{2^j\g(\pi^\mN_{i^*/2},b),\idiv{n-1}{b}} \\  
    &\ge\min\joukko{2^j\idiv{n}{kb},\idiv{n-1}{b}},   
\end{align*}
for $b\in\joukko{1,\ldots,n-1}$ and $j\in\NN$.  Furthermore, for
$b\le n/k$, we get 
$2^{i^*/2}\lattia{n/(kb)}\ge 2k\cdot(n/(2kb))=n/b$ and therefore
$\g(\pi^\mN,b)=\idiv{n-1}{b}$.  If $b>n/k\ge k$, then
$\idiv{n-1}{b}\le n/k$, so $\g(\pi^\mN,a)=\idiv{n-1}{a}$
holds for $a=\idiv{n-1}{b}$.  However, clearly 
$\idiv{n-1}{a}\ge b$, so by the symmetry property $\g(\pi^\mN,a)\ge b$
implies $\g(\pi^\mN,b)\ge a=\idiv{n-1}{b}$.
Hence, $\g(\pi^\mN,b)=\idiv{n-1}{b}$.
\end{proof}

\subsection*{Pseudoloose sets}

Now that possibilities for extending a partial multiplication have been
explored, it is relatively easy to extract some sufficient properties
of cardinality quantifiers that can be used for defining multiplication.
To start with, we need some notation.

For a set~$S\osaj\NN$, we define mappings~$\delta_S\colon S\to\ZZ_+\cup\{\infty\}$
and
$\g_S\colon \ZZ_+\times\ZZ_+\to\NN$ as follows:  For each $n\in S$,
$\delta_S(n)=m-n$ where $m$ is the least number in $S$ which is bigger
than~$n$ if such an~$m$ exist, otherwise $\delta_S(n)=\infty$.  
For every $n,t\in\ZZ_+$, let $S_{n,t}$ be the set
of $m\in S$ with $\delta_S(m)\ge t$ and $m+t<n$.
Put $\g_S(n,t)=|S_{n,t}|$. For $\kh\colon\NN\to\joukko{0,1}$
and a binary word $w\in\joukko{0,1}^s$, $s\in\NN$, put
   \[
T^\kh_w=\joukko{m\in\NN 
                 \mid \forall i\in\joukko{0,\ldots,s-1}\,(\kh(m+i)=w(i))},
   \]
i.e., $T^\kh_w$ is the set of occurrences of the word~$w$ in the infinite
word~$\kh$.

\begin{definition}\label{loose}
Let $S\osaj\NN$, $n\in\NN$ and $\e>0$.  The set $S$ is 
\emph{$\e$-loose relative to $n$} if there exists $t\in\ZZ_+$
satisfying the inequalities $t\g_S(n,t)\ge\e n$
and $n^\e\le t\le n^{1-\e}$.
$S$ is \emph{loose} if there is $\e>0$ such that
for almost all $n\in\ZZ_+$ it holds that $S$ is $\e$-loose 
relative to~$n$.
$S$ is \emph{pseudoloose} if for some $\e>0$, we have that
for almost all $n\in\ZZ_+$ there exists a word $w\in\joukko{0,1}^s$
such that $s\le n^{1-\e}$ and $T^{\kh_S}_w$ is $\e$-loose
relative to~$n$, where $\kh_S$ is the characteristic function of~$S$.
\end{definition}

We note that if $S$ is loose, then it is obviously pseudoloose.
The complement $\NN\setminus S$ of a loose set~$S$ is also pseudoloose,
but in general, not loose, e.g., if $S$ does not contain
consecutive natural numbers.

We are now quite close to fulfilling our goal.
Two steps remain:  To implement the combinatorial idea
explained in the beginning of this section,
we need to be able to compare 
sizes of sets, i.e., we need
$\FO_\le(\kI)\le\FO_\le(\kC_S)$. Then we need to show that a sufficient part
of the definition
of multiplication can be formalized in $\FO_\le(\kC_S)$.
For the first subgoal, we need some combinatorial analysis of
the set~$S$.  Write 
$[a,b]_\NN=[a,b]\cap\NN=\joukko{n\in\NN \mid a\le n\le b}$ for $a,b\in\NN$
(and similarly for other types of intervals).

\begin{lemma} \label{estim-goof}
For infinite $S\osaj\NN$ and $n,t\in\ZZ_+$, we have
$t(\g_S(n,t)-1)\le 2f_S(n)$ or $t\le\oo_S(n)$ where
$f_S$ and $\oo_S$ are as in Theorem~\ref{Idef}.
\end{lemma}

\begin{proof}
Suppose $t>\oo_S(n)$.  Let $S_{n,t}$ be as above and
$\Delta=[f_S(n),n-f_S(n)]_\NN$.
By definition of $\oo_S$, it holds that $S$ is periodic
on the interval~$[f_S(n)-\oo_S(n),n-f_S(n)+\oo_S(n)]$
with period $\oo_S(n)$, 
so for all $a\in S\cap\Delta$ we have $a+\oo_S(n)\in S$
and thus $\delta_S(a)\le\oo_S(n)<t$, which implies $a\not\in S_{n,t}$.
Hence, $\Delta\cap S_{n,t}=\emptyset$.  Moreover, 
${[}a,a+t{[}_\NN\cap\Delta=\emptyset$ for $a\in S_{n,t}$ except
for one possible exception which, if it exists, is the unique~$a\in S_{n,t}$
for which $f_S(n)\in{[}a,a+t{[}_\NN$.  Let $s$ be this unique
exception if it exists, otherwise pick any~$s$.
Because the union 
$A=\bigcup_{a\in S_{n,t}\smallsetminus\joukko{s}}{[}a,a+t{[}_\NN$
is disjoint and included in $\joukko{0,\ldots,n-1}\smallsetminus\Delta$,
we get 
$t(\g_S(n,t)-1)=|A|
\le |\joukko{0,\ldots,n-1}\smallsetminus\Delta|\le 2f_S(n)$.
\end{proof}

\begin{lemma} \label{estim-foo}
Let $S\osaj\NN$, $w\in\joukko{0,1}^s$ and $T=T^{\kh_S}_w$.
Then for every $n\in\NN$ we have $f_T(n)\le f_S(n)+s$.
Moreover, $\oo_T(n)\le\oo_S(n)$ provided that $3f_S(n)+3s\le n$.  
\end{lemma}

\begin{proof}
For the first inequality, we may assume that 
$f_S(n)+s\le\katto{\frac{n}{2}}$ , but then it is obvious that
the periodicity of~$S$ on $[f_S(n)-\oo_S(n),n-f_S(n)+\oo_S(n)]_\NN$ implies
the periodicity of~$T$ on 
$\Delta=[f_S(n)+s-\oo_S(n),n-f_S(n)-s+\oo_S(n)]_\NN$, 
both with period~$\oo_S(n)$, and the inequality follows.

For the second inequality, suppose towards contradiction that
$\oo_T(n)>\oo_S(n)$, but $3f_S(n)+3s\le n$.
Besides $\Delta$ as above, consider 
$\Delta^*=[f_T(n)-\oo_T(n),n-f_T(n)+\oo_T(n)]_\NN$.
We already know that $T$ is periodic on~$\Delta$ with period~$\oo_S(n)$,
so the reason for $\oo_T(n)>\oo_S(n)$ must be that 
$\Delta\subsetneq\Delta^*$ and on $\Delta^*$, $T$ has not got
period $\oo_S(n)$.  However, 
$|\Delta|=n-2f_S(n)-2s+2\oo_S(n)+1\ge f_S(n)+s+2\oo_S(n)+1
>f_T(n)+\oo_S(n)\ge\oo_T(n)+\oo_S(n)$,
so for $x,x+\oo_S(n)\in\Delta^*$, we may pick $y\in\Delta$
such that also $y+\oo_S(n)\in\Delta$ and $y\equiv x\Pmod{\oo_T(n)}$.
As $T$ has period $\oo_T(n)$ on $\Delta^*$ and period $\oo_S(n)$
on $\Delta$, we get $x\in T$ iff $y\in T$ iff $y+\oo_S(n)\in T$
iff $x+\oo_S(n)\in T$.  Hence, $T$ has period~$\oo_S(n)$ on~$\Delta^*$
in contradiction with the minimality of $\oo_T(n)$.
\end{proof}

\begin{proposition}
If $S$ is pseudoloose, then $\FO_\le(\kI)\le\FO_\le(\kC_S)$.  
\end{proposition}

\begin{proof} 
We need to verify the non-periodicity condition of Theorem~\ref{Idef}.
By definition of pseudolooseness, there exists $\e>0$ such that
for almost all $n\in\ZZ_+$ there is $w\in\joukko{0,1}^s$
with $s\le n^{1-\e}$ such that $T^{\kh_S}_w$ is $\e$-loose
relative to~$n$.  Fix $n$ for a moment and put $T=T^{\kh_S}_w$.
Choose $t\in\ZZ_+$ such that $t\g_T(n,t)\ge\e n$ and
$n^\e\le t\le n^{1-\e}$.   By Lemma~\ref{estim-goof},
we get $f_T(n)\ge\frac12 t(\g_T(n,t)-1)\ge \frac12(\e n-n^{1-\e})$
or $\oo_T(n)\ge t\ge n^\e$.  In the former case, we have
$f_S(n)\ge f_T(n)-s\ge \frac12\e n-\frac32 n^{1-\e}$, by 
Lemma~\ref{estim-foo}.  In the latter case, we have
$\oo_S(n)\ge\oo_T(n)\ge n^\e$ or $f_S(n)>\frac13 n-s\ge\frac13 n-n^{1-\e}$,
again by Lemma~\ref{estim-foo}.

These estimates imply that for almost all $n\in\NN$, we have
$f_S(n)\ge\frac14\e n$ or $\oo_S(n)\ge n^\e$.
Choosing $l\in\ZZ_+$ with $l\e\ge 1$ and $k=4l$
we get 
$$
  kf_S(n)\oo_S(n)^l\ge\max\joukko{4lf_S(n),\oo_S(n)^l}
  \ge\min\joukko{(l\e)n,n^{l\e}}\ge n,
$$
for almost all $n\in\ZZ_+$.
Hence, the nonperiodicity condition of Theorem~\ref{Idef}
follows.  
\end{proof}

\begin{theorem} \label{tcchar}
Let $S\osaj\NN$ be a pseudoloose set. Then
   \[
\FO_{\le}(\kC_S)\ge\TC.
   \]  
\end{theorem}

\begin{proof}
With eye on the characterization~$\FO_{\{+,\times\}}(\kI)\equiv\TC$,
we need to show that the H\"artig quantifier, addition and
multiplication are definable in $\FO_\le(\kC_S)$.
The definability of~$\kI$ and addition follows from
the preceding proposition and H\"artig's observation (see \eqref{Hartig's trick}).
It remains to be shown that multiplication is definable.

We proceed rather informally.  Let $\mN$ be an ordered 
$\joukko{\le}$-structure with $n=\card(\mN)$.  For notational
simplicity, we assume
that $\Dom(\mN)=\joukko{0,\ldots,n-1}$ and $\le^\mN$ is
the natural order.  Let $\alpha(x,y,z)$ be a formula in $\FO_\le(\kC_S)$
defining addition (restricted to the appropriate domain)
on ordered structures.  Write $S(x)$ for $\kC_S t(t<x)$;
then $S^\mN=S\cap\joukko{0,\ldots,n-1}$.  Similarly, there is a formula
$\tau(x,\vq)$ of $\FO_\le(\kC_S)$ such that if $w\in\joukko{0,1}^s$,
then $\tau$ defines $T^{\kh_S}_w\cap\joukko{0,\ldots,n-s-1}$
with parameters.  The parameters~$\vq$ simply point to
some occurrence of $w$ if such exists, and addition is used
to find other occurrences.  

Given $t\in\ZZ_+$, consider now $T=T^{\kh_S}_w$ and 
$T_{n,t}=\joukko{m\in T \mid \delta_T(m)\ge t,m+t<n}$.
It is straightforward to write $\theta(x,t,\vp)$ of $\FO_\le(\kC_S)$
such that with appropriate parameters (interpretations of $t$ and $\vp$), 
$\theta$ defines $T'=T_{n,t}\cap\joukko{0,\ldots,n-s-t}$. Finally,
we write a formula $\nu(x,y,z)$ of $\FO_\le(\kC_S)$ defining a partial
multiplication by expressing the following:  Let $a,b,c\in\Dom(\mN)$.
Then $\mN\models\nu[a/k,b/y,c/z]$ iff there are parameters
such that $T'$ as above has an initial segment~$T''$ of size $|T''|=b$,
the union $C=\bigcup_{d\in T''}{[}d,d+a-1{[}_\NN$ is disjoint
and $c=|C|$.  Note that here we use, again, definability of addition
and $\kI$ for defining the appropriate intervals and for comparing
sizes.
 
Next, we estimate the extent of the partial multiplication $\nu^\mN$.
Given $t\in\ZZ_+$, we note that we may compute the product $ab$
where $a,b\in\Dom(\mN)$  with the idea presented above 
provided that $a<t$ and $b\le|T'|$.
Therefore
   \[
\g(\nu^\mN,t)\ge |T'|\ge |T_{n,t}|-\left(\frac{s+t}{t}+1\right)
=\g_T(n,t)-s/t-2.
   \]
As $S$ is pseudoloose,  there is $\e>0$ such that for almost all
$n\in\ZZ_+$ we can choose parameters $s,t\in\ZZ_+$ so that for
$\joukko{\le}$-structure~$\mN$ with $n=\card(\mN)$,
   \[
\g(\nu^\mN,t)\ge\g_T(n,t)-s/t-2
\ge \frac{\e n-n^{1-\e}}{t}-2
\ge \frac{\e n-2n^{1-\e}}{t}
   \]
and $n^\e\le t\le n^{1-\e}$.  Fixing any $k\in\ZZ_+$ with $k\e>1$,
this means that for almost all $n\in\ZZ_+$ there is $t\in\ZZ_+$
such that $n^{1/k}\le t\le n^{1-1/k}/k$ and $kt\g(\nu^\mN,t)\ge n$.
In other words, the assumptions of Proposition~\ref{mulext} are
satisfied for $\nu^\mN$ as the partial multiplication relation. 
Consequently, the formula $\pi^*=\pi(\alpha/A,\nu/M)$
of $\FO_\le(\kC_S)$ defines multiplication in sufficiently large
ordered structures.  Once we fix the finitely many exceptions
of small structures, we are done.
\end{proof}

\subsection*{Analyzing pseudolooseness}

We end this section by working out some concrete examples of our theory.
We first study what are the relevant features of polynomials of
degree at least two so that their ranges are loose sets.
We then show also that looseness is a robust notion in the sense that
it is preserved under quite substantial perturbations of the set.

For clarity of exposition, some elementary analysis is employed.
Recall that a mapping $f\colon I\to\RR$ with $I\osaj\RR$ an interval
is \emph{convex}, if for all $a,b\in I$ and $t\in\av{0,1}$
we have that $f(ta+(1-t)b)\le tf(a)+(1-t)f(b)$.
To facilitate the computation of definite integrals, we write
$\pri{a}{b}F(x)$ for $F(b)-F(a)$.

We say that a statement holds \emph{for sufficiently large~$x$}, 
if there is $M\in\RR$ such that it holds for all $x>M$.
Let $\F$ be the class of strictly increasing and convex mappings
$f\colon \sav{0,\infty}\to\sav{0,\infty}$ with a continuous
derivative on $\av{0,\infty}$ satisfying the following condition:
There are $\alpha,\beta>1$ such that for sufficiently large~$x$
we have 
\begin{equation} \label{alfabeeta}
\alpha\le\frac{xf'(x)}{f(x)}\le\beta.  
\end{equation}

\begin{lemma}
Let $f\in\F$ and $\alpha,\beta>1$ be as in the Inequality~\ref{alfabeeta}.
Then there are positive constants $c_0$, $c_1$, $c_2$  and $c_3$
such that for sufficiently large~$x$ the following hold:

\begin{enumerate}[(a)]
\item
$c_0x^\alpha\le f(x)\le c_1x^\beta$.

\item
There is $\xi\in\av{0,x}$ such that $(x-\xi)f'(\xi)\ge\frac{f(x)}{\beta+1}$.

\item
The number~$\xi$ above can be chosen so that 
$c_2f(x)^{1-1/\alpha}\le f'(\xi)\le c_3f(x)^{1-1/\beta}$.

\end{enumerate}

\end{lemma}

\begin{proof}
\begin{enumerate}[(a)]
\item
Suppose Inequality~\ref{alfabeeta} holds for all $x\ge M>0$.
Then for all $x\ge M$ we have
\begin{align*}
  & \int^x_M \frac{\alpha}{t}\dt
     \le \int^x_M \frac{f'(t)}{f(t)}\dt
     \le  \int^x_M \frac{\beta}{t}\dt\\
\intertext{giving}
  &\pri{M}{x} \alpha\ln t
     \le \pri{M}{x} \ln f(t)
     \le \pri{M}{x} \beta\ln t\\
\intertext{and}
  &\ln \frac{x^\alpha}{M^\alpha}=\alpha \ln x-\alpha \ln M
      \le \ln \frac{f(x)}{f(M)}   
      \le \ln \frac{x^\beta}{M^\beta}.
\end{align*}
Since the natural logarithm is strictly increasing, we get
$c_0x^\alpha\le f(x)\le c_1x^\beta$ with
positive $c_0=f(M)/M^\alpha$ and $c_1=f(M)/M^\beta$.

\item
Consider $I=\int^x_M tf'(t)\dt$ where $x\ge M$.
  On the one hand, integrating by parts we get
   \[
I=\pri{M}{x}tf(t)-\int^x_M f(t)\dt
=xf(x)-Mf(M)-\int^x_M f(t)\dt.
   \] 
On the other hand, the inequality~\ref{alfabeeta} implies
   \[
I\le\int^x_M \beta f(t)\dt=\beta\int^x_M f(t)\dt.
   \]
Combining these we get $(1+1/\beta)I\le xf(x)-Mf(M)\le xf(x)$.
Hence, $I\le \frac{xf(x)}{1+1/\beta}$.  Consequently,
\begin{align*}
   &\int^x_M (x-t)f'(t)\dt=  x\int^x_M f'(t)\dt - I\\ 
   &\ge xf(x)-xf(M)-\frac{xf(x)}{1+1/\beta}
    =\frac{xf(x)}{\beta+1}-xf(M).
\end{align*}
By case~a, say, we have $Mf(x)\ge (\beta+1)xf(M)$
for sufficiently large~$x$.  That implies
   \[
\int^x_M (x-t)f'(t)\dt \ge \frac{xf(x)}{\beta+1}-xf(M)
\ge\frac{xf(x)-Mf(x)}{\beta+1}
=(x-M)\frac{f(x)}{\beta+1}.
   \]
By the intermediate value theorem for integrals
(note that the integrand is continuous) there is $\xi\in\av{M,x}$
such that
   \[
(x-\xi)f'(x)=\frac{1}{x-M}\int^x_M (x-t)f'(t)\dt\ge\frac{f(x)}{\beta+1}.
   \]

\item  
The number~$\xi$ above is known to satisfy $\xi>M>0$ so that
   \[
f'(\xi)\le\frac{\beta f(\xi)}{\xi} 
\le \beta {c_1}^{1/\beta}\frac{f(\xi)}{f(\xi)^{1/\beta}}
=c_3 f(\xi)^{1-1/\beta}
\le c_3 f(x)^{1-1/\beta}
   \]
by Inequality~\ref{alfabeeta} and case~a 
where $c_3=\beta {c_1}^{1/\beta}$.
Furthermore, by cases~a and~b
   \[
f'(\xi)\ge\frac{1}{\beta+1}\,\frac{f(x)}{x-\xi}
\ge\frac{1}{\beta+1}\,\frac{f(x)}{x}
\ge c_2 f(x)^{1-1/\alpha}
   \]
with $c_2={c_0}^{1/\alpha}/(\beta+1)>0$.
\end{enumerate}
\end{proof}

\begin{proposition}
For $f\in\F$, the set $S=\joukko{\lattia{f(x)} \mid x\in\NN}\osaj\NN$
is loose.  
\end{proposition}

\begin{proof}
Choose $\alpha,\beta>1$ so that the Inequality~\ref{alfabeeta}
holds for sufficiently large~$x$.  Fix also $M>\beta(\beta+1)$
so that the results of the preceding lemma hold for all $x\ge M$.
We may also assume that $f'(x)\ge 1$, for $x\ge M$.

Let $x\in\NN$, $x\ge M$.  By convexity of~$f$, we have
$f(x+1)-f(x)\ge f'(x)\ge 1$ so that $\lattia{f(x+1)}>\lattia{f(x)}$
and we get the estimate
   \[
\delta_S(\lattia{f(x)})=\lattia{f(x+1)}-\lattia{f(x)}
>(f(x+1)-1)-f(x)\ge f'(x)-1,
   \]  
or better, $\delta_S(\lattia{f(x)})\ge\katto{f'(x)-1}$.

Let $n\in\ZZ_+$, $n\ge f(M)$.  As $\lim_{x\to\infty}f(x)=\infty$
and $f$ is continuous, we can pick $x\ge M$ such that $f(x)=n$.
By cases b and c of the preceding lemma, choose $\xi\in\av{0,x}$
such that $(x-\xi)f'(\xi)\ge \frac{f(x)}{\beta+1}=\frac{n}{\beta+1}$
and $c_2 n^{1-1/\alpha}\le f'(\xi)\le c_3 n^{1-1/\beta}$.
Consider $y=\katto{\xi}\in\NN$ and $t=\katto{f'(\xi)}-1\in\NN$.
For sufficiently large~$n$, 
we have $\xi>M$ and that $t$ is positive.
As $f$ is convex, the derivative~$f'$ is increasing and for each
$z\in\NN$, $y\le z<\lattia{x}-1$, we have 
$\delta_S(\lattia{f(z)})\ge\katto{f'(z)-1}\ge\katto{f'(y)-1}\ge t$
and 
$\lattia{f(z)}+t\le\lattia{f(z)}+\delta_S(\lattia{f(z)})
\le\lattia{f(\lattia{x}-1)}<f(x)=n$.
Hence, $\g_S(n,t)\ge \lattia{x}-y-1\ge x-y-2$.

We now get
\begin{align*}
t\g_S(n,t) 
   &\ge t(x-y-2) \ge (x-\xi-3)(f'(\xi)-1) \\  
   &\ge (x-\xi)f'(\xi)-x-3f'(x) \\
   &\ge (x-\xi)f'(\xi)-x-\frac{3\beta f(x)}{x} \\
   &\ge \frac{n}{\beta+1}-\left(\frac{n}{c_0}\right)^{1/\alpha}
                         -\frac{3\beta n}{(n/c_1)^{1/\beta}} \\
   &=\frac{n}{\beta+1}-c_4 n^{1/\alpha}-c_5 n^{1-1/\beta}
\end{align*}
with appropriate positive constants $c_4$ and $c_5$.
For almost all $n\in\ZZ_+$ we thus have
   \[
t\g_S(n,t)\ge \frac{n}{\beta+1}-c_4 n^{1/\alpha}-c_5 n^{1-1/\beta}
          \ge \frac{n}{\beta+2}.
   \]

Fix $\e>0$ such that 
$\e<\min\joukko{\frac{1}{\beta+2},1-\frac{1}{\alpha}}$.
Then for almost every $n\in\ZZ_+$ there is $t\in\ZZ_+$ such that
$t\g_S(n,t)\ge\frac{n}{\beta+2}\ge \e n$ and
$n^\e\le c_2 n^{1-1/\alpha}-1 \le f'(\xi)-1\le t\le f'(\xi) 
\le c_3 n^{1-1/\beta}\le n^{1-\e}$.
Hence, for almost all $n\in\ZZ_+$ the set~$S$ is $\e$-loose
relative to~$n$. 
\end{proof}

\begin{corollary}
\begin{enumerate}[(a)]
\item Let  $P\colon \NN\to\NN$ be a polynomial function with  
coefficients in $\NN$ and $\deg(P)=k\ge 2$, and put $S=\rg(P)$. Then
$\FO_{\le}(\kC_S) \equiv \TC$.
\item Let $r>1$ be a real and 
$S_r=\joukko{\lattia{x^r}\mid x\in \NN  }$. Then 
$\FO_{\le}(\kC_{S_r}) \ge \TC$.
\end{enumerate}

\end{corollary}

\begin{proof}
\begin{enumerate}[(a)]
\item
Clearly, $\FO_\le(\kC_S)\le\TC$.
Let $f\colon \sav{0,\infty}\to \sav{0,\infty}$ be the canonical
extension of~$P$.  Then 
   \[
\lim_{x\to\infty}\frac{xf'(x)}{f(x)}=k
   \]
so that for sufficiently large~$x$, we have
$1<k-\frac12<\frac{xf'(x)}{f(x)}<k+\frac12$.
Other conditions for~$f$ being obvious, we see that $f\in\F$
and therefore $S=\joukko{\lattia{f(x)} \mid x\in\NN}=\rg(P)$
is loose.  Hence by Theorem~\ref{tcchar}, $\FO_{\le}(\kC_S)\ge\TC$.

\item
Let $g\colon \sav{0,\infty}\to \sav{0,\infty}$, $g(x)=x^r$.
Then for all $x>0$, we have $\frac{xg'(x)}{g(x)}=r>1$,
which implies $g\in\F$ and the result. \qedhere
\end{enumerate}
\end{proof}

Observe that for $r'>r>1$, the sets $S_r$ and $S_{r'}$ are distinct.
This implies that there are uncountably many loose sets.
Obviously, most of them are not computable, and as such quite
uninteresting from the point of view of descriptive complexity theory.

Next we show that looseness is preserved under certain modifications.
The following easy result is for the record.

\begin{lemma} \label{loosediff}
Let $S,S'\osaj\NN$ be infinite sets such that $S\triangle S'$
is finite.  Then $S$ is loose iff $S'$ is loose.  
\end{lemma}

\begin{proof}
The main point here is that $\g_S-\g_{S'}$ is obviously bounded;
pick $M\in\NN$ such that for every $n,t\in\ZZ_+$, we have
that $\left|\g_S(n,t)-\g_{S'}(n,t)\right|\le M$.  By symmetry,
we need to show the implication only in one direction, so suppose
$S$ is loose, i.e., for some $\e>0$, we have that for almost every
$n\in\ZZ_+$, $S$ is $\e$-loose relative to~$n$.  Pick any $\e'>0$
with $\e'<\e$.  Then for almost every $n\in\ZZ_+$ there is $t\in\ZZ_+$
such that $n^{\e'}\le n^\e\le t\le n^{1-\e}\le n^{1-\e'}$ and
$t\g_{S'}(n,t)\ge t(\g_S(n,t)-M)\ge \e n-Mn^{1-\e}\ge \e'n$.
Therefore, $S'$ is loose.
\end{proof}

We employ a commonly used concept in metric space theory:
Let $S,S'\osaj\NN$ be infinite and $\eta\ge 1$.  
A mapping $f\colon S\to S'$
is \emph{$\eta$-bi-Lipschitz} if  for all $m,n\in S$
   \[
\eta^{-1}|m-n|\le \left|f(m)-f(n)\right|\le \eta|m-n|.
   \]
If $f\colon S\to S'$ is a strictly increasing bijection, then
this condition simplifies:  $f$ is a $\eta$-bi-Lipschitz mapping
iff for every $n\in S$, we have
   \[
\eta^{-1}\delta_S(n)\le\delta_{S'}(f(n))\le\eta\delta_S(n). 
   \]
$f$ is \emph{bi-Lipschitz} if there exists $\eta\ge 1$ such that
$f$ is $\eta$-bi-Lipschitz.

The key point is the following:

\begin{lemma}
Let $S,S'\osaj\NN$ be infinite sets and $f\colon S\to S'$
a strictly increasing $\eta$-bi-Lipschitz bijection with $\eta\ge 1$.
Suppose $0\in S\cap S'$.  Then for every $n,t\in\ZZ_+$ 
we have $\g_{S'}(n,t)\ge\g_S(\katto{n/\eta},\katto{\eta t})$. 
\end{lemma}

\begin{proof}
Write $m=\katto{n/\eta}, s=\katto{\eta t}\in\ZZ_+$.  Let $S_{m,s}$
and $S'_{n,t}$ be the appropriate sets witnessing $\g_S(m,s)=|S_{m,s}|$ 
and $\g_{S'}(n,t)=|S'_{n,t}|$.  For $x\in S_{m,s}$, we have $f(x)\in S'$
and $f(x)+t=|f(x)-f(0)|+t\le\eta|x|+t\le\eta(x+s)<n$,
as $f(0)=0$ and $x+s<m=\katto{n/\eta}$ implies $x+s<n/\eta$.
Furthermore, we have 
$\delta_{S'}(f(x))\ge\eta^{-1}\delta_S(x)\ge\eta^{-1}s\ge t$,
so $f(x)\in S'_{n,t}$.  Hence, $f\raj S_{m,s}\colon S_{m,s}\to S'_{n,t}$
is an injection, completing the proof. 
\end{proof}

\begin{proposition}\label{loose-biLip}
Let $S,S'\osaj\NN$ be infinite sets. Suppose that there are
$S_0\osaj S$ and $S'_0\osaj S'$ such that $S\smallsetminus S_0$
and $S'\smallsetminus S'_0$ are finite and there is
a strictly increasing bi-Lipschitz bijection $f\colon S_0\to S'_0$.
Then $S$ is loose iff $S'$ is loose.
\end{proposition}

\begin{proof}
Put $S_1=\joukko{0}\cup(S_0\smallsetminus\joukko{\min S_0})$  
and $S'_1=\joukko{0}\cup(S'_0\smallsetminus\joukko{\min S'_0})$.
Then it is easy to see that 
$f_1=(f\raj(S_1\smallsetminus\joukko{0}))\cup\joukko{(0,0)}$
is a strictly increasing bi-Lipschitz mapping (with some
inferior Lipschitz coefficient $\eta$).
By symmetry, it suffices to
consider only one direction of the equivalence, so suppose $S$
is loose.  By Lemma~\ref{loosediff}, $S_1$ is also loose.
Fix $\e>0$ such that for almost every $n\in\ZZ_+$, $S_1$ is 
$\e$-loose relative to~$n$.

Put $\e'=\e/(2\eta^2)>0$.    For almost every $n\in\ZZ_+$ , 
$n^{\e-\e'}\ge 2\eta^{1+\e}$ and $S_1$ is $\e$-loose relative
to $\katto{n/\eta}$.  Suppose $t$ witnesses the latter.
Then $u=\lattia{t/\eta}\in\NN$ satisfies
   \[
u\ge \lattia{\katto{n/\eta}^\e/\eta}
   \ge \lattia{\frac{n^\e}{\eta^{1+\e}}}
   \ge \lattia{2n^{\e'}} \ge n^{\e'},
   \] 
$u\le t\le n^{1-\e'}$ and by the preceding lemma,
\begin{align*}
u\g_{S'_1}(n,u) 
   &\ge u\g_{S_1}(\katto{n/\eta},t) \\     
   &\ge \frac{u\e \katto{n/\eta}}{t} 
         = \e\frac{\lattia{t/\eta}\katto{n/\eta}}{t} \\
   &\ge \frac{\e}{2\eta^2}\,n=\e'n,
\end{align*}
as $\g_S$ is decreasing with respect to the second variable
and $t\ge\eta$.  Hence, $S'_1$ is $\e'$-loose relative to~$n$.
It follows that $S'_1$ is loose implying that $S'$ is loose, too.
\end{proof}

\begin{example}
Consider $S=\joukko{k^2+\lattia{\lb k}\cdot (-1)^k \mid k\in\ZZ_+}$.
We compare this with 
$S'=\joukko{k^2 \mid k\in\ZZ_+}=\joukko{(k+1)^2 \mid k\in\NN}$  
which is loose as a range of a quadratic polynomial.
Let $f\colon S'\to S$ 
be the unique strictly increasing bijection.
Obviously
   \[
\lim_{k\to\infty}\frac{\delta_S(f(k^2))}{\delta_{S'}(k^2)}=1,
   \]
so $f$ is a bi-Lipschitz mapping and also $S$ is loose.
It is straightforward to show that $\FO_\le(\kC_S)\le\TC$.
Putting these together, we get $\FO_\le(\kC_S)\equiv \TC$.
\end{example}

As the last enterprise of the section, we show that
$E=\{2^n \mid n\in\NN\}$ is not pseudoloose.  Recall that by
Lemma~\ref{Idefex}~(\ref{Idef-E}),   the equicardinality
quantifier $\kI$ is definable in~$\FO_{\le}(\kC_E)$, so that
addition is definable in~$\FO_{\le}(\kC_E)$. The non-pseudolooseness
of~$E$ leaves it as an open question if multiplication is
also definable in that logic.

We first show that in a more general setting, pseudolooseness
reduces to looseness.

\begin{proposition}
Suppose $S$ is the range of a strictly increasing infinite sequences
with strictly increasing differences of consecutive elements,
i.e., $S=\{ a_n \mid n\in\NN \}$ where $a_{k+2}-a_{k+1}>a_{k+1}-a_k>0$,
for each $k\in\NN$.  Then $S$ is pseudoloose if and only if
$S$ is loose.
\end{proposition}

\begin{proof}
If $S$ is loose, then it is obviously also pseudoloose, so we concentrate
on the reverse direction. So suppose $S$ is pseudoloose.
Fix $\e>0$ and $n_0\in\ZZ_+$ with $\e n_0^\e>2$ such that
for $n\in\ZZ_+$ with $n\ge n_0$ there exists a word $w\in\joukko{0,1}^s$
such that $s\le n^{1-\e}$ and $T^{\kh_S}_w$ is $\e$-loose
relative to~$n$.  We aim to show that for almost all $n\in\ZZ_+$,
$S$ is $\e/2$-loose relative to~$n$.  Given $n$, choose a word $w$
that satisfies the given conditions.  First note that $w$ cannot
contain two or more 1's.  Indeed, for such a word~$w$, we have
that $T=T^{\kh_S}_w$ is empty or a singleton due to the strictly
increasing differences of consecutive elements in~$S$.
This implies that we would have
$$
   t\g_T(n,t)\le t\cdot 1\le n^{1-\e} < \frac12 \cdot \e n_0^\e n^{1-\e}
   \le\frac12 \e n,  
$$
contrary to the choice of~$w$.  Therefore, $w$ contains at most one~1.

Consider now the case when $w$ contains exactly one~1.  
Then $T=T^{\kh_S}_w\osaj S-k$ where $k<s$ is the index of the only~1.
A moment's reflection show that 
$$
  t\g_T(n,t) \le t\g_S(n,t)+k,
$$
for any $t\in\ZZ_+$.  As $T$ is $\e$-loose relative to~$n$, there
exists $t\in\ZZ_+$ such that $t\g_T(n,t)\ge\e n$
and $n^\e\le t\le n^{1-\e}$.  Note now that
$k\le s\le n^{1-\e}\le \frac12\e n$.  Combining the estimates,
we have
$$
  \e n \le t\g_T(n,t) \le t\g_S(n,t)+k \le t\g_S(n,t)+\frac12\e n,
$$
which implies $t\g_S(n,t)\ge \frac12\e n$.  Trivially,
$n^{\e/2}\le t\le n^{1-\e/2}$, so we see that $S$ is $\e/2$-loose
relative to~$n$.

It remains to handle the case when $w$ consists only of 0's.
Then $T$ contains many consecutive 1's, but for $t\ge 2$,
we have that $T_{n,t}\osaj S-s$.  The rest of the proof is
similar to the previous case.
\end{proof}

\begin{example} \label{nonpseudoloose}
Condider now the set $E=\{2^n \mid n\in\NN\}$. We aim to show that
$E$~is not pseudoloose.  In view of the previous proposition,
it is enough to prove that $E$ is not loose.  Note that
$$
  | E \cap \{ 0,\ldots,n-1 \}| = \lattia{\lb(n-1)} +1, 
$$
for every $n\in\ZZ_+$.  Observe also that for every $t\in\ZZ_+$,
we have $E_{n,t}\osaj E$ where $E_{n,t}$ is the set
of $m\in E$ with $\delta_E(m)\ge t$ and $m+t<n$, similarly as
in the Definition~\ref{loose}.  This implies that 
$\g(n,t)=|E_{n,t}|\le \lattia{\lb(n-1)} +1 \le \katto{\lb(n)}$.

Fix now $\e>0$. Given any $t\le n^{1-\e}$, we have that
$$
  t\g(n,t) \le n^{1-\e}\katto{\lb(n)}<\e n,
$$
if $n$ is sufficiently large.  This means that $E$ is not  
$\e$-loose relative to~$n$, for almost all $n\in\ZZ_+$.
Hence, $E$ is neither loose nor pseudoloose.
\end{example}

\section{Built-in relations}\label{Builtin}

In this section, we develop the basic theory of logics with built-in
relations.  Built-in relations have an important role in descriptive
complexity theory: all the known logical characterizations of
complexity classes below $\PTIME$ require the presence of built-in
linear order. In particular this holds for the famous characterization
of $\PTIME$ in terms of least fixed-point logic by Immerman \cite{I86}
and Vardi \cite{V82}.  However, as far as we know, the abstract notion
of a logic with built-in relations has not been studied systematically
earlier.  Here we will adapt the notions of generalized quantifiers,
semiregularity and regularity of logics from abstract model theory to
the setting with built-in relations.

The standard approach to logics with built-in relations (see, e.g.,
\cite{EF,I}) is that the name $S$ of each built-in relation is
included in the vocabulary $\tau$ considered, but the interpretation
of $S$ in $\tau$-structures is restricted to be the intended one
(e.g., for $S={\le}$, a linear order of the universe).  Since we want
to compare the expressive power of logics with different sets of
built-in relations, we need to solve the following problem: How can a
class of models that is defined in a logic $\LL_\B$ with a set $\B$ of
built-in relations be definable in another logic
$\widetilde\LL_{\widetilde\B}$ with a different set $\widetilde\B$ of
built-in relations?  In the usual approach this can be done only by
treating the names of built-in relations in $\B$ as ordinary relation
symbols in formulas of $\widetilde\LL_{\widetilde\B}$, and forcing the
intended interpretations for these by sentences of
$\widetilde\LL_{\widetilde\B}$.  We introduce a different solution to
this problem: we replace the usual notion of model with built-in
relations with a uniform one which is independent of the set $\B$ of
built-in relations available in the logic considered. Thus, in our
approach the built-in relations are not included in the vocabulary of
the structures; they are rather regarded as part of the logical
vocabulary in the same spirit as is done with the identity relation.

All models we consider in this and the next section are assumed to be finite. That is, for any model $\mM$ considered in these sections, $\Dom(\mM)$ is assumed to be a finite nonempty set, and the vocabulary $\tau$ of $\mM$ is assumed to be finite, as well.

\subsection*{Logics with built-in relations}

We call any relation $S$ on the natural numbers $\NN$ a {\em
numerical relation}. If $\mM$ is a finite structure of some vocabulary $\tau$,
then any set $\B$ of numerical relations can be interpreted
as {\em built-in relations} on $\mM$ by fixing a bijection
$f$ between $\Dom(\mM)$ and $n=\{0,\ldots,n-1\}\in\NN$: 
each $S\in \B$ is considered as a relation symbol that is not in $\tau$,
and the interpretation of $S$ with respect to the
bijection $f$ is defined to be the relation $S^f=\{\va\mid f\va\in S\}$.
Conversely, if
$\mM^*=(\mM,(S^{\mM^*})_{S\in\B})$ is an expansion of~$\mM$, where
the interpretations $S^{\mM^*}$ are regarded as built-in relations, 
then there is a bijection $f\colon \Dom(\mM)\to n$
such that $S^{\mM^*}=S^f$ for all $S\in\B$. Thus, the models we will
consider in the context of logics with built-in relations are
pairs of ordinary models $\mM$ and bijections $\Dom(\mM)\to |\Dom(\mM)|$.

\begin{definition}\label{br-models}
Let $\tau$ be a relational vocabulary. If $\mM$ is a
$\tau$-model, and $f\colon \Dom(\mM)\to |\Dom(\mM)|$ is a bijection, then
$\mM^f=(\mM,f)$ is a {\em $\tau$-model with potential built-in relations};
more briefly, we say that $\mM^f$ is a {\em $\BR$-$\tau$-model}.
We write $\Str_\BR(\tau)$ for the class of all $\BR$-$\tau$-models.
\end{definition}

The most important built-in relations for the topic of this paper are
linear order and the arithmetic relations for addition and
multiplication. We denote the numerical relation corresponding to
built-in linear order simply by $\le$, for addition by $\bplus$ and
for multiplication by $\btimes$. Thus,
$\bplus=\{(i,j,k)\in\NN^3\mid i+j=k\}$ and
$\btimes=\{(i,j,k)\in\NN^3\mid i\cdot j=k\}$, and for any $\BR$-model
$\mM^f$ and elements $a,b,c\in\Dom(\mM)$, we have
$a\le^f b\iff f(a)\le f(b)$ and
$(a,b,c)\in\bplus^f\iff f(a)+f(b)=f(c)$ (and similarly for
$\btimes$).

The notion of isomorphism can be extended to $\BR$-models
in a natural way as follows: a $\BR$-{\em isomorphism} between
two $\BR$-$\tau$-models $(\mM,f)$ and  $(\mN,g)$ is a bijection
$h\colon \Dom(\mM)\to \Dom(\mN)$ such that $h$ is an isomorphism
between the $\tau$-models $\mM$ and $\mN$ and $f=g\circ h$.
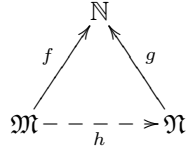
   $$
\xy
0;<10mm,0mm>:<0mm,15mm>::
(-1,0)*+<1.5mm>{\mM}="M";  (1,0)*+<1.5mm>{\mN}="N";
(0,1)*+<1.5mm>{\NN}="NN"; 
{\ar^f "M";"NN"}; {\ar@{-->}_h "M";"N"}; {\ar_g "N";"NN"};
\endxy
   $$
The $\BR$-models $(\mM,f)$ and $(\mN,g)$ are isomorphic,
$(\mM,f)\cong (\mN,g)$, if there exists a $\BR$-isomorphism 
between them. 

It is easy to see that if $h\colon (\mM,f)\to (\mN,g)$ is a $\BR$-isomorphism, 
then for any set $\B$ of numerical relations, $h$ is an isomorphism between 
the corresponding models $(\mM,(S^f)_{S\in\B})$ and 
$(\mN,(S^g)_{S\in\B})$ with built-in relations. On the other hand, 
if $h\colon (\mM,(S^f)_{S\in\B})\to(\mN,(S^g)_{S\in\B})$ is an isomorphism, 
then there is an automorphism $e$ of $(\mN,(S^g)_{S\in\B})$ such that
$h$ is a $\BR$-isomorphism between $(\mM,f)$ and $(\mN,g\circ e)$.
In particular, if $\B$ contains the linear order $\le$ of natural numbers
(or some other relation such that $(\mM,(S^f)_{S\in\B})$ cannot have
non-trivial automorphisms), then $(\mM,(S^f)_{S\in\B})
\cong(\mN,(S^g)_{S\in\B})$ implies $(\mM,f)\cong (\mN,g)$.

Note that for every $\BR$-model $(\mM,f)$ there is a canonical 
representative of its isomorphism type: Let $f\mM$ be the image 
of $\mM$ under the bijection $f$, and let $g$ be the identity
function of $\Dom(f\mM)$. Then $(f\mM,g)$ is the unique
$\BR$-model with domain $n=|\Dom(\mM)|$ and bijection $g={\rm id}_n$ 
which is $\BR$-isomorphic with $(\mM,f)$.

Let $\B$ be a set of numerical relations, and let $\LL$ be an abstract
logic (we refer to \cite{E} for the definition of abstract logics).
We define now the corresponding logic $\LL_{\B}$ {\em with
built-in relations} $\B$. As mentioned before Definition~\ref{br-models}, we use the numerical relations in $\B$ as relation symbols in formulas of $\LL_{\B}$.

\begin{definition}
For each vocabulary $\tau$, the set $\LL_{\B}[\tau]$ of $\tau$-sentences of
$\LL_\B$ is $\LL[\tau\cup\B]$. Furthermore, the
truth relation is inherited from $\LL$: for all $(\mM,f)\in\Str_\BR(\tau)$
and $\phi\in\LL_\B[\tau]$,
$$
   (\mM,f)\models_{\LL_\B}\phi\iff
   (\mM,(S^f)_{S\in\B})\models_{\LL}\phi.
$$
In the sequel we will call logics with built-in relations $\B$
just $\B${\em -logics}.
\end{definition}

To simplify the presentation, we will use the convention that a
{\em formula} is a sentence in an expanded vocabulary, i.e.,  we regard
variables as constant symbols, and an $\LL_\B[\tau]$-formula $\phi$ with
free variables $x_0\ldots,x_{k-1}$ is just a sentence in
$\LL_\B[\tau\cup\{x_0,\ldots,x_{k-1}\}]$. Let $(\mM,f)$ be a
$\BR$-$\tau$-model, and let $\phi$ be an $\LL_\B[\tau]$-formula with
free variables $x_0,\ldots,x_{k-1}$.
We write $(\mM,f)\models\phi[a_0/x_0,\ldots,a_{k-1}/x_{k-1}]$, or more
briefly $(\mM,f)\models\phi[\va /\vx]$, if $(\mM^+,f)\models\phi$, where
$\mM^+$ is the $\tau\cup\{x_0,\ldots,x_{k-1}\}$-expansion of $\mM$ with
$a_i=x_i^{\mM^+}$ for each $i< k$. Furthermore, we use the notation
$$
   \phi^{\mM,f} =\{\va\in\Dom(\mM)^k\mid (\mM,f)\models\phi[\va/\vx]\}
$$
for the {\em relation defined by the formula} $\phi$ in the model $(\mM,f)$.

Each sentence $\phi$ of a $\B$-logic $\LL_\B$ {\em defines} a class
$K_\phi$ of $\BR$-models:
$$
   K_\phi=\{(\mM,f)\in\Str_\BR(\tau)\mid (\mM,f)\models\phi\},
$$
where $\tau$ is the vocabulary of $\phi$. The expressive power of
$\LL_\B$ is determined by the collection of classes that are definable by
$\LL_\B$-sentences.
The comparison between a $\B$-logic $\LL_\B$ and a
${\widetilde{\B}}$-logic $\widetilde{\LL}_{\widetilde{\B}}$
is defined in terms of their expressive power in the usual way:

\begin{definition}\label{comparison}
$\LL_\B$ is {\em at most as strong}
as $\widetilde{\LL}_{\widetilde{\B}}$, in symbols
$\LL_\B\le\widetilde{\LL}_{\widetilde{\B}}$,
if every class $K\subseteq\Str_{\BR}(\tau)$ which
is definable in $\LL_\B$ is also definable in
$\widetilde{\LL}_{\widetilde{\B}}$. Furthermore, we say that
$\LL_\B$ is {\em strictly weaker} than $\widetilde{\LL}_{\widetilde{\B}}$,
$\LL_\B<\widetilde{\LL}_{\widetilde{\B}}$, if
$\LL_\B\le\widetilde{\LL}_{\widetilde{\B}}$ and
$\widetilde{\LL}_{\widetilde{\B}}\not\le\LL_\B$.
If $\LL_\B\le\widetilde{\LL}_{\widetilde{\B}}$ and
$\widetilde{\LL}_{\widetilde{\B}}\le\LL_\B$, we write
$\LL_\B\equiv\widetilde{\LL}_{\widetilde{\B}}$, and say that
$\LL_\B$ and $\widetilde{\LL}_{\widetilde{\B}}$ are \emph{equivalent}.
\end{definition}

As an example, consider first order logic with two sets $\B$
and $\widetilde{\B}$ of built-in relations. If
$\FO_\B\le\FO_{\widetilde{\B}}$, then for each relation
$S\in\B$ there is a formula $\phi_S\in\LL_{\widetilde{\B}}[\emptyset]$
such that $\phi_S$ defines $S$ in every $\BR$-model
$(\mM,f)$: $\phi_S^{\mM,f}=S^f$. This is because
the class $K_S=\{(\mM,f)\in\Str_{\BR}(\{x_0,\ldots,x_{k-1}\})
\mid (x_0^\mM,\ldots,x_{k-1}^\mM)\in S^f\}$ is trivially definable
in $\FO_\B$, so by assumption, it is also definable
in $\LL_{\widetilde{\B}}$. The converse is also
true: if all relations $S$ in $\B$ are definable by formulas $\phi_S$
of $\FO_{\widetilde{\B}}$ in this way, then $\FO_\B\le\FO_{\widetilde{\B}}$.
This is easy to prove by using the fact that $\FO$ allows substituting
relation symbols by formulas (see \cite{E}).

\begin{example}\label{order-sum}
For any $a,b\in\NN$, $a\le b$ holds if and only if there is $c\in\NN$ such that $a+c=b$. Thus the $\FO_+$-formula $\phi_\le(x,y):=\exists z\,\bplus(x,z,y)$ defines the order $\le$ on all $\BR$-models. Hence, we see that $\FO_\le\le\FO_+$.
\end{example}

In the literature, a class of structures with built-in linear order is said to be order-invariant if membership of a model $\mM$ in the class does not depend on the order $\le^\mM$.
We adopt this terminology to the corresponding notion in our framework: a class $K\subseteq\Str_\BR(\tau)$ is {\em order-invariant}
if the equivalence 
$$
  (\mM,f)\in K\iff (\mM,g)\in K
$$
holds for all $\tau$-models $\mM$ and all bijections 
$f,g\colon \Dom(\mM)\to |\Dom(\mM)|$. 

If $\phi$ is a sentence of 
some $\B$-logic $\LL_\B$ such that none of the relations 
$S\in \B$ occurs in $\phi$, then the class $K_\phi$
defined by $\phi$ is obviously order-invariant. Thus, we see
that Definition~\ref{comparison} extends the usual notion 
of comparing the expressive power of logics in the sense that 
$\LL\le \widetilde\LL$ if and only if $\LL_\B\le\widetilde\LL_\B$
with $\B=\emptyset$.

\subsection*{Generalized quantifiers with built-in relations}

The notion of generalized quantifier needs to be adapted to the
framework of built-in relations. While ordinary generalized quantifiers, like the Härtig quantifier, can be
used without problems in the context of built-in relations, the
characterization of semiregularity for $\B$-logics (see
Proposition \ref{semireg}) requires the notion of {\em quantifier
with potential built-in relations}, or more briefly,
$\BR$-{\em quantifier}.

\begin{definition}\label{b-quantifier}
Let $K_{Q}\subseteq\Str_\BR(\tau)$ be a class of $\BR$-$\tau$-models
which is closed under $\BR$-isomorphisms.
Then the corresponding $\BR$-quantifier $Q$ is a syntactic operator
which can be used for binding tuples of variables $\vx_R$ and
formulas $\psi_R$, $R\in\tau$, of some vocabulary $\sigma$ to obtain 
a new formula $Q\, (\vx_R\psi_R)_{R\in\tau}$ of vocabulary $\sigma$. 
The semantics of $Q$ is given by the clause
\[
   (\mM,f)\models Q\, (\vx_R\psi_R)_{R\in\tau}\iff
   (\Dom(\mM),(\psi_R^{\mM,f})_{R\in\tau},f) \in K_Q.
\]
Here the components of the tuple $\vx_R$ are assumed to be distinct, and the length of $\vx_R$ is assumed to be equal to the arity of $R$ for each $R\in\tau$.

The extension $\LL_\B(\Q)$ 
of a $\B$-logic $\LL_\B$ with a
set $\Q$ of $\BR$-quantifiers is obtained by adding the new formula formation rule
\begin{center}
    if $\psi_R$, $R\in\tau$, are formulas, then $Q\, (\vx_R\psi_R)_{R\in\tau}$ is a formula
\end{center}
 with the semantics above
for each $Q\in\Q$. Here we assume that the syntax and semantics of $\LL$ are given by a collection of formula formation rules and corresponding semantic clauses. This holds for all the concrete logics we consider in the paper.

We say that a $\BR$-quantifier $Q$ {\em is definable in} a $\B$-logic
$\LL_\B$ if its defining class $K_Q$ is definable in $\LL_\B$.
Furthermore, we say that $Q$ is {\em order-invariant} if $K_Q$
is order-invariant.
\end{definition}

It is worth noting that, by replacing $\BR$-$\tau$-models by ordinary  $\tau$-models in 
Definition~\ref{b-quantifier}, we recover the definition of an ordinary generalized quantifier 
of vocabulary~$\tau$   \cite{Lin}. 
The crucial difference to the standard definition of generalized
quantifiers is that even though the built-in relations in $\B$ are present
in the class $K_Q$ via the bijections $f$, they are not explicitly bound by $Q$ 
in the formula $Q\, (\vx_R\psi_R)_{R\in\tau}$.
This is important, since
otherwise definable relations could be substituted in place of them,
which would be against the whole idea of built-in relations being part of logical vocabulary.

In typical examples the vocabulary $\tau$ of a $\BR$-quantifier $Q$ is usually given in a more explicit form as $\tau=\{R_0,\ldots,R_{k-1}\}$. In such a case we use the notation $Q\,\vx_0,\ldots,\vx_{k-1}\,(\psi_0,\ldots,\psi_{k-1})$ instead of $Q\, (\vx_R\psi_R)_{R\in\tau}$.


\begin{example} \label{bqex}
Consider the class $K_Q\subseteq\Str_{\BR}(\joukko{P_a,P_b,P_c})$
that consists of all $\BR$-models
$(\mM,f)$ such that $|P_a^\mM|=|P_b^\mM|$, 
$\Dom(\mM)=P_a^\mM\cup P_b^\mM\cup P_c^\mM$, and $u <^f v <^f w$ whenever $u\in P_a^\mM$,
$v\in P_b^\mM$ and $w\in P_c^\mM$.
In the language of directed graphs, the sentence
    \[
Q\, x,y,z\;(\phi(x),\psi(y),\lnot\phi(z)\land\lnot\psi(z))
    \]
with $\phi(x)\sij \lnot\exists z E(z,x)$
and $\psi(x)\sij \lnot\exists z E(x,z)$
expresses that, in the built-in order, we have first the sources
of the directed graph, then the sinks (without overlap),
then the rest of the vertices, and there are equally many
sources and sinks.  Note that this example is not order-invariant.
\end{example}

The structures in $K_Q$ of the preceding example can be interpreted
as word models
corresponding to the language $L=\joukko{a^mb^mc^n \mid m,n\in\NN}$.
More generally, given any language $L\subseteq\Sigma^*$, we can define
the corresponding {\em language quantifier}: $Q_L$ is the
$\BR$-quantifier of the vocabulary $\tau_\Sigma=\{P_a\mid
a\in\Sigma\}$ such that $K_{Q_L}$ consists of all models
$(\mM,f)\in\Str_{\BR}(\tau_\Sigma)$ which encode strings in $L$.

\begin{figure}
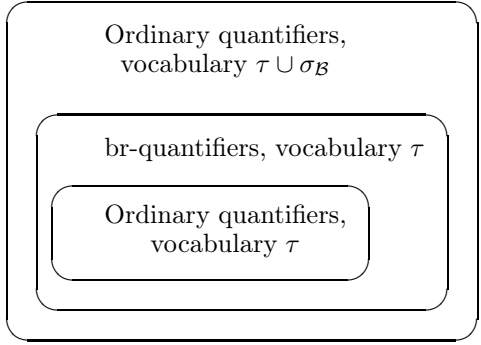
 
    $$\xy
0;<15mm,0mm>:
<-4mm,-2.5mm> !LU+<-2mm,4mm>="1",
0*+<2mm>!LU\txt{Ordinary quantifiers,\\vocabulary $\tau$} !RD+<2mm,-2mm>="1t",
"1"."1t" *\frm<3mm>{-}, !LD+<-2mm,-2mm>="1b" !LD+<-4mm,-4mm>="1c";
(0,0.6)*+<2mm>!LU\txt{$\BR$-quantifiers, vocabulary $\tau$},
!RU+<2mm,2mm>="2t", !RU+<4mm,4mm>="2u",
"1b"+<0mm,-2mm>."2t" *\frm<3mm>{-};
(0,1.6)*+<2mm>!LU\txt{Ordinary quantifiers,\\
                       vocabulary $\tau\cup\sigma_\B$}
!RU+<2mm,2mm>="3",
"1c"+<0mm,-2mm>."2u"."3" *\frm<3mm>{-}
    \endxy$$
\caption{Comparison between ordinary and $\BR$-quantifiers}
\label{qcomp}
\end{figure}

Figure~\ref{qcomp} compares the expressive capabilities of ordinary
quantifiers and quantifiers with built-in relations schematically.
Let us first explore the step from ordinary quantifiers to
$\BR$-quantifiers with the same vocabulary.
Any ordinary quantifier $Q$ can be interpreted as an order-invariant
$\BR$-quantifier $Q_\BR$ simply by expanding the models in $K_Q$ by all
possible ways to $\BR$-models: let
$$
   K_{Q_\BR}=\{(\mM,f)\mid \mM\in K_Q,\;
   f\colon \Dom(\mM)\to |\Dom(\mM)|\hbox{ is a bijection}\}.
$$
For example, if
$Q$ is the divisibility quantifier $\kD_2$, then we obtain in this way
the language quantifier $(\kD_2)_\BR=Q_L$, where $L=\{w\in\{a,b\}^*\mid
|\{i\mid w_i=a\}|\hbox{ is even}\}$.

Note that since the class $K_{Q_\BR}$ is order-invariant, given any 
model  $\mM$ and formulas $\psi_R$ for ${R\in\tau}$, the equivalence 
$$
  \mM\models Q\,(\vx_R\psi_R)_{R\in\tau}\iff 
  (\mM,f)\models Q_\BR\,(\vx_R \psi_R)_{R\in\tau}
$$
holds for all bijections $f\colon \Dom(\mM)\to |\Dom(\mM)|$.
Thus, we say that an \emph{ordinary quantifier $Q$ is 
definable} in a $\B$-logic $\LL_\B$ if $Q_\BR$ is definable in
$\LL_\B$.

On the other hand, for any set $\B$ of built-in relations, each
$\BR$-quantifier~$Q$ with vocabulary~$\tau$
can be lifted to an ordinary quantifier~$Q^\B$ with
vocabulary~$\tau\cup\sigma_\B$ where $\sigma_\B:=\{\tilde S\mid S\in \B\}$ is
a disjoint copy\footnote{Note that the symbols $\tilde S$ are not part of the logical vocabulary, unlike the original symbols $S\in\B$. This is why we use the notation $\sigma_\B$ instead of  $\tilde\B$.} of $\B$.
The defining class $K_{Q^\B}$ of the \emph{lift}~$Q^\B$ is
$$
   J_{Q,\B}=\joukko{(\mM,(S^f)_{\tilde S\in\sigma_\B}) \mid (\mM,f)\in K_Q}
$$
 Here naturally $(\mM,(S^f)_{\tilde S\in\sigma_\B})$ denotes the $\sigma_\B$-expansion $\mM^*$ of $\mM$ such that $\tilde S^{\mM^*}=S^f$ for each $\tilde S\in\sigma_\B$.
However, syntactically the lift behaves differently from the original
$\BR$-quantifier~$Q$, since the built-in relations in $\B$,
or any other definable relations taking over the role
of built-in relations, are explicitly bound whenever $Q^\B$ is applied.  

Naturally, this has an impact
on the semantics, too.  The sentences
$Q\, (\vx_R\psi_R)_{R\in\tau}$ and
$Q^\B\, (\vx_R\psi_R,\vy_S S(\vy_S))_{R\in\tau,S\in\B}$
are equivalent, but the syntax of $Q^\B$ allows also
more general sentences of form
$Q^\B\, (\vx_R\psi_R,\vy_S \theta_S(\vy_S))_{R\in\tau,S\in\B}$.

Note the difference between the steps $Q\mapsto Q_\BR$
and $Q\mapsto Q^\B$.  On one hand, if $Q$ is an ordinary quantifier,
then $\FO_\B(Q)\equiv \FO_\B(Q_\BR)$, but it is clear that there
are plenty of $\BR$-quantifiers that are not equivalent to
ordinary quantifiers in this way
(see Example~\ref{bqex}).  On the other hand, if $Q$ is a
$\BR$-quantifier, then $\FO_\B(Q)\le\FO_\B(Q^\B)$, but in general,
the $\B$-logics are not equivalent.
We shall return to this point in the following subsection,
Example~\ref{bqexcont}.

\subsection*{Semiregular $\B$-logics}

In the next two subsections, we will introduce the key
notions of regularity for logics with built-in relations.
These notions are obtained by adapting the usual framework
of abstract logic (see \cite{E}).

\begin{definition}
We say that a $\B$-logic $\LL_\B$
is {\em semiregular}, if for every vocabulary $\tau$, all atomic
formulas in $\tau$ are expressible in $\LL_{\B}[\tau]$, all relations
in $\B$ are definable in $\LL_{\B}[\tau]$, $\LL_{\B}$ is closed under
Boolean operations and first order quantification, and $\LL_\B$ is
{\em closed under substitution}:
\begin{itemize}
\item[(s)] If $\psi_R(\vx_R)$ are $\LL_\B[\sigma]$-formulas with
$|\vx_R|=\arity(R)$ for each $R\in\tau$, and $\phi$ is an
$\LL_\B[\tau]$-sentence, then there is an $\LL_\B[\sigma]$-sentence
$\theta$ such that
\[
   (\mM,f)\models\theta\iff
   (\Dom(\mM), (\psi_R^{\mM,f})_{R\in\tau},f)\models\phi
\]
holds for all $(\mM,f)\in\Str_\BR(\sigma)$.
\end{itemize}
\end{definition}

Note that substituting relations defined by formulas in place of the
built-in relations in $\B$ is not included in this definition. This is because
of the very idea of built-in relations: they are thought as {\em constants}
that are always present in the models considered.

For logics of the form $\FO_\B(\Q)$, substitution of formulas in place
of relations can be defined syntactically:
let $\phi[(\psi_R/R)_{R\in\tau}]$ be the sentence obtained by
replacing each occurrence $R(\vy)$ of $R$ in $\phi$ by the formula
$\psi_R(\vy)$. It straighforward to show that the equivalence
in (s) always holds for $\theta\sij\phi[(\psi_R/R)_{R\in\tau}]$.
Thus, we have

\begin{lemma}\label{foq-semireg}
$\FO_{\B}(\Q)$ is semiregular for any set $\B$ of built-in relations and
any class $\Q$ of $\BR$-quantifiers.
\end{lemma}

Any semiregular $\B$-logic $\LL_\B$ is closed under quantification
with respect to any $\LL_\B$-definable quantifier. We prove next a useful
formulation of this principle.

\begin{lemma}\label{q-semireg}
Let $\B$ and $\widetilde{\B}$ be sets of numerical relations,
let $\Q$ be a set of $\BR$-quantifiers, and let $\LL_{\widetilde{\B}}$
be a semiregular $\widetilde{\B}$-logic. If all relations $S\in\B$ and
all quantifiers $Q\in\Q$ are definable in $\LL_{\widetilde{\B}}$, then
$\FO_\B(\Q)\le\LL_{\widetilde{\B}}$.
\end{lemma}

\begin{proof}
We prove by induction
that for every formula $\phi\in\FO_{\B}(\Q)[\sigma]$ there is an equivalent
formula $\theta\in\LL_{\widetilde{\B}}[\sigma]$, i.e.,
$\phi^{\mM,f}=\theta^{\mM,f}$, for every $(\mM,f)\in\Str_{\BR}(\sigma)$.
The claim is true for
atomic formulas in the vocabulary $\sigma$ by the definition of
semiregularity; for atomic formulas $S(\vx)$ with $S\in\B$, we
use the assumption that $S$ is definable in $\LL_{\widetilde{\B}}$.
The induction steps corresponding
to connectives and existential quantifier go through since,
by semiregularity, $\LL_{\widetilde{\B}}$ is closed
with respect to Boolean operations and first order quantification.

Consider finally the step corresponding to a quantifier $Q\in\Q$: let
$\phi$ be the formula $Q\, (\vx_R\psi_R)_{R\in\tau_Q}$, and assume that for each
$\psi_R$, $R\in\tau_Q$, there is an equivalent formula
$\eta_R\in\LL_{\widetilde{\B}}$. By our assumption on $\Q$,
there is a sentence $\chi\in\LL_{\widetilde{\B}}[\tau_Q]$ such
that $K_{Q}=\{(\mN,g)\in\Str_{\BR}(\tau_Q)\mid(\mN,g)\models\chi\}$. Since
$\LL_{\widetilde{\B}}$ is closed under substitution,
there is a sentence $\theta\in\LL_{\widetilde{\B}}[\sigma]$
such that for every $\BR$-$\sigma$-model $(\mM,f)$,
\[
   (\mM,f)\models\theta\iff
   (\Dom(\mM), (\eta_R^{\mM,f})_{R\in\tau_Q},f)\models\chi.
\]
On the other hand, as $\eta_R^{\mM,f}=\psi_R^{\mM,f}$ for each 
$R\in\tau$, we have
\[
   (\Dom(\mM), (\eta_R^{\mM,f})_{R\in\tau_Q},f)\models\chi
   \iff(\mM,f)\models Q\, (\vx_R\psi_R)_{R\in\tau_Q}.
\]
Thus we conclude that $\phi$ is equivalent with $\theta$.
\end{proof}

Without built-in relations, semiregularity
of a logic can be characterized in terms of generalized quantifiers:
$\LL$ is semiregular if and only if there is a class $\Q$ of
quantifiers such that $\LL\equiv\FO(\Q)$
(see, e.g., \cite{E}).
This characterization
remains valid in the framework of logics with built-in relations, once
we replace ordinary quantifiers with the appropriate $\BR$-quantifiers:

\begin{proposition}\label{semireg}
A logic $\LL_\B$ with built-in relations is semiregular if and only if there
is a class $\Q$ of $\BR$-quantifiers such that
$\LL_{\B}\equiv\FO_{\B}(\Q)$.
\end{proposition}

\begin{proof}
We prove the implication from left to right; the other implication follows
directly from Lemma \ref{foq-semireg}.
Thus, assume that $\LL_\B$ is semiregular. We will show that
$\LL_\B\equiv\FO_\B(\Q)$, where $\Q$ is the class of all
$\BR$-quantifiers which are definable in $\LL_\B$.
Note first that $\LL_{\B}\le\FO_{\B}(\Q)$. Indeed, if $K$ is class of
$\BR$-structures which is definable in $\LL_\B$, then $K=K_Q$ for
a quantifier $Q\in\Q$, whence $K$ is trivially definable in $\FO_{\B}(\Q)$.
On the other hand, since all the relations in $\B$ and all the quantifiers
in $\Q$ are definable in $\LL_\B$, we have $\FO_{\B}(\Q)\le\LL_{\B}$
by Lemma \ref{q-semireg}.
\end{proof}

\begin{example} \label{bqexcont}
Consider the $\BR$-quantifier~$Q$ of Example~\ref{bqex}.
We show that $\FO_\le(Q)$ is strictly weaker than $\FO_\le(Q^\le)$.  First we note that $Q$
is definable in~$\FO_+$ since if $(\mM,f)\in\Str_\BR(\joukko{P_a,P_b,P_c})$
and $P^\mM_a$, $P^\mM_b$  and $P^\mM_c$ follow each other in the
required order (which is clearly $\FO_\le$-expressible), then
$|P^\mM_a|=|P^\mM_b|$ if and only if $v=u+^f u+^f 1$ where
$u=\max(P^\mM_a)$ and $v=\max(P^\mM_b)$.
Since $\FO_+$ is semiregular and $\le$ is definable in $\FO_+$ (see Example~\ref{order-sum}), we have $\FO_\le(Q)\le\FO_+$.

The strength of the lift is that the order may be tailored according to
our needs. Put
\begin{align*}
x\le' y\sij
     & \bigl( (U(x)\ekv U(y)) \land (V(x)\ekv V(y))\land x\le y \bigr) \\
     & \lor \bigl( U(x)\land \lnot V(x) \land
                     \lnot (U(y)\land \lnot V(y))\bigr) \\
     & \lor \bigl( \lnot U(x)\land V(x)
                     \land (U(y)\ekv V(y))\bigr) \\
     & \lor \bigl( U(x)\land V(x) \land \lnot U(y) \land \lnot V(y)\bigr).
\end{align*}
Thus, in every $(\mM,f)\in\Str_\BR(\joukko{U,V})$ the formula~$x\le' y$
defines a linear order such that $U^\mM\smallsetminus V^\mM$ is an 
initial segment
followed by the interval~$V^\mM\smallsetminus U^\mM$.
Consequently,
\[
   Q^\le\, x,y,z, tu\,(U(x)\land\lnot V(x),V(y)\land\lnot U(y),
   U(z)\ekv V(z), t\le' u)
\]
defines the H\"artig quantifier~$\kI$ in the logic~$\FO_\le(Q^\le)$.
On the other hand,
$Q^\le$ is readily seen to be definable in $\FO_\le(\kI)$,
so $\FO_\le(Q^\le)\equiv\FO_\le(\kI)$.

By Ajtai's result~\cite{A}, parity of sets is not
definable even in $\FO_\bit\ge\FO_+$, so
$\FO_\le(\kD_2)\not\le\FO_\le(Q)$, whereas
$\FO_\le(\kD_2)\le\FO_\le(\kI)\equiv\FO_\le(Q^\le)$.
Hence, we conclude that $\FO_\le(Q)<\FO_\le(Q^\le)$.
\end{example}

In the light of the preceding example, it is interesting to observe that
if we restrict attention to strong enough $\B$-logics, the notion of a 
$\BR$-quantifier
becomes superfluous and $\BR$-quantifiers may be replaced by
ordinary quantifiers, namely by their lifts.

\begin{proposition}
Let $\B$ be a finite set of built-in relations with ${\le}\in\B$,
and let $Q$ be a $\BR$-quantifier.
Then $\FO_\B(Q,\kI)\equiv\FO_\B(Q^\B,\kI)$ holds for the lift $Q^\B$
of~$Q$.
\end{proposition}

\begin{proof}
As $\FO_\B(Q,\kI)\le\FO_\B(Q^\B,\kI)$ and $\FO_\B(Q,\kI)$
is semiregular, it suffices to show that $Q^\B$ is definable
in $\FO_\B(Q,\kI)$. To do this, we need to find a
sentence $\phi$ of $\FO_\B(Q,\kI)$ such that the equivalence
$(\mN,f)\models\phi\iff \mN\in K_{Q^\B}$
holds for all
$\BR$-models $(\mN,f)$ of vocabulary $\tau\cup\sigma_\B$.

Thus, consider a $\BR$-model
$(\mN,f)\in\Str_\BR(\tau\cup\sigma_\B)$ with $|\Dom(\mN)|=n$.
Now $\mN\in K_{Q^\B}$ if and only if
there is a bijection
$g\colon\Dom(\mN)\to n$ such that
$\mN=(\mM,(S^g)_{\tilde S\in\sigma_\B})$ and $(\mM,g)\in K_Q$,
where $\mM=\mN\restriction\tau$.
We will give $\FO_\B(Q,\kI)$-sentences $\alpha$, $\beta$ and $\gamma$
such that
\begin{itemize}
\item[(1)] if $(\mN,f)\models\alpha$, then there is an
$\FO_\B(Q,\kI)$-definable permutation $h\colon \Dom(\mN)\to \Dom(\mN)$
such that ${\widetilde\le^\mN}={\le^g}$ for the bijection $g=f\circ h$;
\item[(2)] if $(\mN,f)\models\alpha\land\beta$, then
$\mN=(\mM,(S^g)_{\tilde S\in\sigma_\B})$, where $\mM=\mN\restriction\tau$;
\item[(3)] if $(\mN,f)\models\alpha\land\beta$, then
$(\mN,f)\models\gamma\iff(\mM,g)\in K_Q$.
\end{itemize}

Let $\alpha\in\FO(\{\widetilde\le\})$
express that $\widetilde\le$ is a linear order of the domain, and let
    \[
\eta(x,y)\sij \kI\,t,u\;(t\;{\widetilde\le}\; x,u\le y).
    \]
If $(\mN,f)\models\alpha$, then
$h:=\eta^{\mN,f}$ is an isomorphism $(\Dom(\mN),\widetilde\le^\mN)
\cong (\Dom(\mN),\le^f)$. Clearly this means that ${\widetilde\le^\mN}={\le^g}$,
where $g\colon \Dom(\mN)\to n$ is the bijection $f\circ h$.
Thus, condition (1) is satisfied.

For condition (2), we need a sentence $\beta$ expressing that
${\tilde S}^\mN=S^g$ for all $S\in\B$. Since $g= f\circ h$, this holds
if and only if the image of ${\tilde S}^\mN$ under $h$ coincides
with $S^f$. Thus, for each $S\in\B$, we let $\theta_S$ be
the $\FO_\B(\kI)[\sigma_\B]$-formula that defines the $h$-image
of $\tilde S$. Now we can let $\beta$ be the sentence
$\bigwedge_{S\in\B}\forall\vx_S (\theta_S(\vx_S)\ekv S(\vx_S))$.

Finally, observe that $(\mM,g)=(\mN\restriction \tau,g)\in K_Q$
if and only if $(h\mM,f)\in K_Q$, where $h\mM$ is the image of $\mM$
under $h$. 
Thus, we can satisfy condition (3) by letting $\gamma$ to
be the $\FO_\B(Q,\kI)[\tau\cup\sigma_\B]$-sentence
$Q\,(\vx_R\,\theta_R(\vx_R))_{R\in\tau}$, where $\theta_R$, $R\in\tau$,
is an $\FO_\B(\kI)[\tau\cup\sigma_\B]$-formula defining the
$h$-image of $R$.

To complete the proof, let $\phi:=\alpha\land\beta\land\gamma$.
By conditions (1)--(3), if $(\mN,f)\models\alpha\land\beta$, then
$(\mN,f)\models\phi\iff \mN\in K_{Q^\B}$. On the other hand, if
$(\mN,f)\not\models\alpha\land\beta$, then either $\widetilde\le^\mN$
is not a linear order, or ${\tilde S}^\mN$, $S\in\B$, are not proper
built-in relations, whence $\mN\not\in K_{Q^\B}$. Thus, we conclude that
$\phi$ defines the quantifier $Q^\B$.
\end{proof}

\subsection*{Regular $\B$-logics}


In the usual framework without built-in relations, an abstract logic $\LL$
is said to be regular if it is semiregular, and it is closed under
relativization (see \cite{E}). The latter condition is based on the notion
of {\em relativization} $\mM|U$ of a model $\mM\in\Str(\tau)$ to a non-empty subset
$U\subseteq\Dom(\mM)$ which is defined as follows:
\begin{itemize}
\item $\Dom(\mM|U)=U$,
\item $R^{\mM|U}=R^\mM\cap U^{\arity(R)}$ for each $R\in\tau$.
\end{itemize}
In the context of $\B$-logics, we need to extend this definition
to $\BR$-models $(\mM,f)$. There is a canonical way of doing this:
given a bijection $f\colon \Dom(\mM)\to |\Dom(\mM)|$
and a subset $U\subseteq\Dom(\mM)$, we use the unique order preserving
bijection $U\to |U|$ in the relativized model.

\begin{definition}
The {\em relativization} of a $\BR$-model $(\mM,f)\in\Str_\BR(\tau)$ 
to a subset
$U\subseteq\Dom(\mM)$ is the $\BR$-model $(\mM,f)|U:=(\mM|U,f_U)$, where
$\mM|U\in\Str(\tau)$ is the relativization of $\mM$ to the set $U$,
and $f_U\colon U\to |U|$ is the unique bijection such that for all $a,b\in U$,
$f_U(a)\le f_U(b)\iff f(a)\le f(b)$.
\end{definition}

Note that 
built-in linear order behaves in a relativization
$(\mM|U,f_U)$ in the same way as relations in the vocabulary $\tau$
of the model $\mM$: ${\le^{f_U}}={{\le^f}\cap U^2}$.
However, this is not true for other built-in relations in general.
For example, the restriction $\bit^f\cap U^2$ of the $\bit$-relation
to a set $U$ is usually no more a $\bit$-relation on $U$.

\begin{definition}
We say that a $\B$-logic $\LL_\B$ is {\em regular}
if it is semiregular and it is {\em closed under
relativization}:
\begin{itemize}
\item[(r)] If $\psi(x)$ is an
$\LL_\B[\tau]$-formula with one free variable,
and $\phi$ is an $\LL_\B[\tau]$-sentence, then there is an
$\LL_\B[\tau]$-sentence $\theta$ such that
\[
     (\mM,f)\models\theta\iff (\mM,f)|\psi^{\mM,f}\models\phi
\]
holds for all $(\mM,f)\in\Str_\BR(\tau)$.
\end{itemize}
\end{definition}

In the case of $\FO_\le$, relativization can be defined syntactically
in the same way as for $\FO$ without built-in relations. Thus,
$\FO_\le$ is regular. We will prove a more general result:
any quantifier extension $\FO_\le(\Q)$ of $\FO_\le$ is regular,
provided that all the quantifiers $Q\in\Q$ admit relativization.
A simple sufficient condition for this is universe independence.
A $\BR$-quantifier $Q$ of vocabulary $\tau$ is {\em universe independent} if
\[
(\mM,f)\in K_Q\iff (\mM,f)|U\in K_Q
\]
whenever $(\mM,f)\in\Str_\BR(\tau)$ and $U\subseteq \Dom(\mM)$ is such that
$R^{\mM}\subseteq U^{\arity(R)}$ for each $R\in\tau$.

\begin{proposition}\label{foq-reg}
$\FO_\le(\Q)$ is regular for any class $\Q$ of universe independent
$\BR$-quanti\-fiers.
\end{proposition}

\begin{proof}
The {\em relativization}
$\phi|\psi$ of a formula $\phi\in\FO_\le(\Q)[\tau]$ with respect to
a formula $\psi\in\FO_\le(\Q)[\tau]$
with one free variable is defined inductively as follows:
\begin{itemize}
\item $\phi|\psi\sij\phi$ \quad for atomic $\tau$-formulas $\phi$,
\item $(x\le y)|\psi\sij x\le y$,
\item $(\lnot\phi)|\psi\sij\lnot \,\phi|\psi$,
\item $(\phi\land\theta)|\psi\sij \phi|\psi\land\theta|\psi$,
\item $(\exists y \,\phi)|\psi\sij \exists y\,(\psi(y)\land\phi|\psi)$
\item $(Q\,(\vx_R\eta_R)_{R\in\tau_Q})|\psi\sij
Q\,(\vx_R(\psi^{k_R}(\vx_R)
\land \eta_R|\psi))_{R\in\tau_Q}$ \quad for $Q\in\Q$.
\end{itemize}
Here $k_R=\arity(R)$ and $\psi^{k_R}(\vx_R)$ is the conjunction of
$\psi(x_i)$ over all components $x_i$ of $\vx_R$.
We prove by induction on $\phi$ that for all $\BR$-models
$(\mM,f)$ and all tuples $\va$ of elements in $\psi^{\mM,f}$
\[
(\mM,f)\models(\phi|\psi)[\va/\vx]\iff (\mM,f)|\psi^\mM\models
\phi[\va/\vx].
\]
In particular, if $\phi$ is a sentence, the
equivalence in condition (r) holds for $\theta\sij\phi|\psi$.

The claim for atomic $\tau$-formulas and
the induction steps for connectives and existential quantifier
are proved exactly as for $\FO$ without built-in relations. The
claim for atomic formulas of the form $x\le y$ follows from the
fact that $\le^{(\mM,f)|\psi^{\mM,f}}$ is the restriction of
$\le^{\mM,f}$ to the set $\psi^{\mM,f}$.
Consider then the induction step for a quantifier $Q\in\Q$.
Let $\mM$ and $\va$ be fixed, and assume that
$\phi$ is of the form $Q\,(\vx_R\eta_R)_{R\in\tau_Q}$. Then
$\phi|\psi\sij Q\,(\vx_R\chi_R)_{R\in\tau_Q}$, where
$\chi_R$ is the formula $\psi^{k_R}(\vx_R) \land \eta_R|\psi$
for each $R\in\tau_Q$. To simplify notation, we denote $\mM|\psi^{\mM,f}$
by $\mN$, $f_{\psi^{\mM,f}}$ by $g$, and the expansions of $\mM$ and 
$\mN$ by the constants
$\va$ by $\mM^+$ and $\mN^+$, respectively.
By induction hypothesis, we have
\[
(\mM,f)\models(\eta_R|\psi)[\va/\vx,
\vb/\vx_R]\iff
(\mN,g)\models\eta_R[\va/\vx,\vb/\vx_R]
\]
for all $R\in\tau_Q$ and for all tuples $\vb$ in $\psi^\mM$, whence
\[
\chi_R^{\mM^+,f}=(\psi^{\mM,f})^{k_R}\cap (\eta_R|\psi)^{\mM^+,f}=
\eta_R^{\mN^+,g}
\]
for all $R\in\tau_Q$. Now we get the following chain of equivalences:
\begin{align*}
(\mM,f)\models(\phi|\psi)[\va/\vx]
&\iff (\Dom(\mM),(\chi_R^{\mM^+,f})_{R\in\tau_Q},f)\in K_Q\\
&\iff (\psi^\mM,(\chi_R^{\mM^+,f})_{R\in\tau_Q},g)\in K_Q \\
&\iff (\mN,g)\models\phi[\va/\vx].
\end{align*}
Here the second equivalence is true since $Q$ is universe independent,
and clearly $\chi_R^{\mM^+,f}\subseteq (\psi^{\mM,f})^{k_R}$ for each
$R\in\tau_Q$.
\end{proof}


Note that the proof of Proposition \ref{foq-reg}  cannot be
extended to $\FO_\B(\Q)$ if $\B$ contains a numerical relation
$S$ such that $S^{f_U}\not= S^f\cap U^{\arity(S)}$ for some bijection
$f\colon \Dom(\mM)\to n$ and some subset $U$ of $\Dom(\mM)$.
In fact, as we shall see in Section \ref{CraneBeach}, $\FO_\B$ is
usually not regular if $\B$ contains such a relation.
However, assuming that ${\le}\in\B$, this problem can be avoided by adding
the H\"artig quantifier to $\FO_\B$.


\begin{proposition}\label{luosto1}
(\cite{Lu1})  For any set $\B$ of numerical relations with ${\le}\in\B$
there is a set $\Q_\B$ of universe independent quantifiers such that
$\FO_\B(\kI)\equiv\FO_\le(\kI,\Q_\B)$.
\end{proposition}

\begin{proof}
Let $S\in \B$, and let $\kB_S$ be the universe independent quantifier with
defining class
\[
    K_{\kB_S}=\{\mM\in\Str_\BR(\{P_0,\ldots,P_{k-1}\})
    \mid (|P_0^\mM|-1,\ldots,|P_{k-1}^\mM|-1)\in S\}.
\]
It is easy to see that $\kB_S$ is definable in $\FO_{\{\le,S\}}(\kI)$.
Conversely, the built-in relation $S$ is clearly definable in 
$\FO_\le(\kI,\kB_S)$. The claim follows now from Lemma~\ref{q-semireg}.
\end{proof}

\begin{example}\label{builder}
(a) Clearly all cardinality quantifiers $\kC_S$, the H\"artig quantifier
$\kI$ and the divisibility quantifier $\kD$ are universe independent.
Hence, the logics $\FO_\le(\kC_S)$, $\FO_\le(\kI)$ and $\FO_\le(\kD)$
are regular.

(b) By Propositions~\ref{luosto1} and \ref{foq-reg}, 
the logic $\FO_\B(\kI)$ is regular for 
any set $\B$ of built-in relations
such that ${\le}\in\B$. 
\end{example}

\section{Regular interior and closure of $\AC$}\label{CraneBeach}

There is a notable feature that separates the circuit complexity class
$\AC$ from $\TC$: the former has a definite weakness against ``padding''
of input strings when computing natural cardinality
properties. Indeed, the property of strings being of even length is
trivially in $\AC$, while, by the famous theorem of Ajtai \cite{A} and
Furst, Saxe and Sipser \cite{FSS}, the property of binary strings of
having an even number of $1$'s is not in $\AC$. This weakness of $\AC$
can be formulated in a precise way in terms of the logic
$\FO_{\{+,\times\}}$ capturing it: $\FO_{\{+,\times\}}$ is not closed
under relativization, and so it is not a regular
$\{+,\times\}$-logic. On the other hand, the logics capturing $\TC$,
like $\FO_{\{+,\times\}}(\kI)$, are regular.

Thus, from a logical perspective, we can say that $\TC$ is a better
behaving class than $\AC$. But regularity is also a very natural
requirement from the computational point of view. Just note that a
logic $\LL_\B$ capturing a complexity class $C$ is closed under
substitution if and only if $C$ is closed under composition of
queries. Similarly, the relativization property for $\LL_\B$ translates to
the requirement that $C$ is closed under restricting $C$-computable
queries to $C$-computable subsets of input structures.

In this section, we will study two ways of addressing the weakness of
$\AC$. The first one is to look for a largest possible fragment of
$\AC$ that is regular. This leads us to the notion of regular
interior of a logic. The second alternative is to look for a minimal
regular extension of $\AC$. For this purpose, we adapt the notion of
regular closure used in the area of Abstract Logic to the case of
$\B$-logics.

\subsection*{Regular interior and regular closure}

Let $\B$ be a set of numerical relations, and let $\LL_\B$
be a semiregular logic. Assume further that the order $\le$ is in $\B$
(or it is definable in $\LL_\B$). Then we can show that there exists a largest 
regular logic that is contained in $\LL_\B$  (see Proposition~\ref{RintL}). 
We call it the
{\em regular interior} of $\LL_\B$, and denote it by
$\RI(\LL_\B)$. Regular interior was introduced for logics without built-in
relations in \cite{Lu2}.

The definition of regular interior is based
on the notion of universe independence, which we introduced
in the previous section.

\begin{definition}
Let $\LL_\B$ be a semiregular $\B$-logic such that $\le$
is definable in $\LL_\B$. We set
$\RI(\LL_\B)\sij\FO_{\le}(\Q_u)$, where $\Q_u$ is the class of all universe
independent $\BR$-quantifiers which are definable in $\LL_\B$.
\end{definition}

Before showing that $\RI(\LL_\B)$ has the desired properties, we
introduce an auxiliary notion, and prove a couple of lemmas.

There is a canonical way of obtaining a universe independent quantifier
from any given quantifier: the {\em regularization}
$Q^{\reg}$ of a $\BR$-quantifier $Q$ is the $\BR$-quantifier of vocabulary
$\sigma =\tau_{Q}\cup\{P\}$ with $P$ a new unary relation symbol,
having the defining class
     \[
K_{Q^{\reg}}=\{(\mM,f)\in\Str_\BR(\sigma)\mid ((\mM,f)|P^{\mM})
\upharpoonright\tau_Q\in K_Q\}.
     \]
Here $((\mM,f)|P^{\mM})\upharpoonright\tau_Q$  is the reduct of
$(\mM,f)|P^{\mM}$ to the vocabulary~$\tau_Q$.

\begin{lemma}\label{Qreg}
$Q^{\reg}$ is universe independent for any $\BR$-quantifier $Q$.
\end{lemma}

\begin{proof}
Assume that $(\mM,f)$ is a $\BR$-model
of vocabulary $\tau_Q\cup\{P\}$ and $U$ is a subset of $\Dom(\mM)$
such that $R^\mM\subseteq U^{\arity(R)}$ for all $R\in\tau_Q\cup \{P\}$.
In particular $P^\mM\subseteq U$, whence clearly $((\mM,f)|U)|P^\mM=
(\mM,f)|P^\mM$, and so we have the chain of equivalences
\begin{align*}
(\mM,f)\in K_{Q^{\reg}}
&\iff ((\mM,f)|P^\mM)\upharpoonright\tau_Q\in K_Q \\
&\iff (((\mM,f)|U)|P^\mM)\upharpoonright\tau_Q\in K_Q \\
&\iff (\mM,f)|U\in K_{Q^{\reg}}.
\end{align*}
\end{proof}

It is easy to see that $Q$ is always definable in $\FO_\le(Q^\reg)$: indeed,
for any $\BR$-model $(\mM,f)$, we have $(\mM,f)\in K_Q\iff (\mM,f)\models
Q^{\reg}(\vx_R \psi_R)_{R\in\tau_Q\cup\{P\}}$, where
$\psi_R\sij R(\vx_R)$ for $R\in\tau_Q$ and $\psi_P\sij (x=x)$. The converse
direction does not hold in general, but it becomes true in the context of
a regular $\B$-logic.

\begin{lemma}\label{regdef}
Let $\LL_\B$ be a regular $\B$-logic, and let $Q$ be a $\BR$-quantifier.
If $Q$ is definable in $\LL_\B$, then $Q^\reg$ is also definable in
$\LL_\B$.
\end{lemma}

\begin{proof}
Assume that $\phi$ is a sentence in $\LL_\B[\tau_Q]$ that defines the class
$K_Q$. Since $\LL_\B$ is closed under relativization,
there is an $\LL_\B$-sentence $\theta$ of vocabulary $\tau_Q\cup\{P\}$
such that the equivalence
\[
(\mM,f)\models\theta\iff(\mM,f)|P^\mM\models\phi
\]
holds for all $(\mM,f)\in\Str_\BR(\tau_Q\cup\{P\})$.
Note further that
\[
(\mM,f)|P^\mM\models\phi\iff ((\mM,f)|P^\mM)\upharpoonright\tau_Q\in K_Q,
\]
since $P$ does not occur in the sentence $\phi$. This means that
$\theta$ defines the class $K_{Q^\reg}$.
\end{proof}

Now we are ready to prove that the definition of regular interior
works as intended:

\begin{proposition}\label{RintL}
Let $\LL_\B$ be a semiregular $\B$-logic such that $\le$
is definable in $\LL_\B$. Then
$\RI(\LL_\B)$ is the largest sublogic of $\LL_\B$ that is regular.
\end{proposition}

\begin{proof} Let $\Q_u$ be the class of all universe independent
$\BR$-quantifiers definable in $\LL_\B$.
Since each quantifier in $\Q_u$ and the linear order $\le$ are definable
in $\LL_\B$, and $\LL_\B$ is semiregular, it follows from
Lemma \ref{q-semireg} that $\RI(\LL_\B):=\FO_\le(\Q_u)
\le \LL_\B$. Moreover, $\RI(\LL_\B)$ is regular by Proposition \ref{foq-reg}.

It remains to prove that $\RI(\LL_\B)$ contains all regular sublogics
of $\LL_\B$. Thus, assume that $\widetilde{\LL}_{\widetilde{\B}}\le\LL_\B$
is regular, and let $Q$ be a $\BR$-quantifier which is
definable in $\widetilde{\LL}_{\widetilde{\B}}$. Since
$\widetilde{\LL}_{\widetilde{\B}}$ is regular, $Q^\reg$ is definable
in $\widetilde{\LL}_{\widetilde{\B}}$ by Lemma \ref{regdef}. Furthermore,
since $\widetilde{\LL}_{\widetilde{\B}}\le\LL_\B$, $Q^\reg$ is definable in
$\LL_\B$ as well. By Lemma \ref{Qreg}, the quantifier $Q^\reg$ is universe
independent, whence it is in the class $\Q_u$. As observed above, $Q$
is definable in $\FO_\le(Q^\reg)$, whence we conclude that $Q$ is
definable in $\RI(\LL_\B)$.
\end{proof}

On the other hand,
every semiregular $\B$-logic $\LL_\B$ can be extended to a regular
$\B$-logic. In fact, assuming again that the order $\le$ is
definable in $\LL_\B$,
there is a least regular extension $\RC(\LL_\B)$
of $\LL_\B$, which we call the {\em regular closure} of $\LL_\B$.
The definition of regular closure uses the notion of regularization
of quantifiers:

\begin{definition}
Let $\LL_\B$ be a $\B$-logic such that $\le$ is definable in $\LL_\B$,
and let
$\Q$ be the class of all $\BR$-quantifiers which are definable in
$\LL_\B$.
We set $\RC(\LL_\B)\sij\FO_{\le}(\Q^{\reg})$, where
$\Q^{\reg}=\{Q^{\reg}\mid Q\in\Q\}$.
\end{definition}

It is straightforward to show that $\RC(\LL_\B)=\FO_{\le}(\Q^{\reg})\equiv\FO_{\B}(\Q^{\reg})$, where $\Q^{\reg}$ is as in the definition above.

Just like in the case of ordinary logics and generalized quantifiers
(see \cite{Lu2}),
we can prove that $\RC(\LL_\B)$
is regular, and there is no regular $\B$-logic strictly in-between
$\LL_\B$ and $\RC(\LL_\B)$.

\begin{proposition}
Let $\LL_\B$ be a $\B$-logic such that $\le$ is definable in $\LL_\B$.
Then $\RC(\LL_\B)$ is the least
extension of $\LL_\B$ that is regular.
\end{proposition}

\begin{proof}
Let $\Q$ be the class of $\BR$-quantifiers which are definable in
$\LL_\B$. Since $Q$ is definable in $\FO_\le(Q^\reg)$ for each $Q\in\Q$,
it follows that $\LL_\B\le \FO_\B(\Q^\reg)\equiv\FO_\le(\Q^\reg)$. Thus,  $\RC(\LL_\B)$
is indeed an extension of $\LL_\B$. Furthermore, since all
the quantifiers in $\Q^\reg$ are universe independent, $\RC(\LL_\B)$
is regular by Proposition~\ref{foq-reg}.

To complete the proof, we assume that $\widetilde{\LL}_{\widetilde{\B}}$
is a regular $\widetilde{\B}$-logic
such that $\LL_\B\le\widetilde{\LL}_{\widetilde{\B}}$.
Then each quantifier $Q\in\Q$ is definable in
$\widetilde{\LL}_{\widetilde{\B}}$. Since
$\widetilde{\LL}_{\widetilde{\B}}$ is regular, it follows from
Lemma~\ref{regdef} that all quantifiers $Q^\reg$ in $\Q^\reg$ are also
definable in $\widetilde{\LL}_{\widetilde{\B}}$.
Moreover, since the order $\le$ is definable in $\LL_\B$, it is
also definable in $\widetilde{\LL}_{\widetilde{\B}}$.
Thus, by Lemma \ref{q-semireg} we conclude that
$\FO_\le(\Q^{\reg})\le\widetilde{\LL}_{\widetilde{\B}}$.
\end{proof}

The definition of the regular closure of a logic $\LL_\B$ is quite abstract in the sense that it refers to the collection of all $\BR$-quantifiers definable in $\LL_\B$. In the case $\LL_\B$ is of the form $\FO_{\{\le,S\}}$, where $S$ is unary, we can give a simple concrete characterization for $\RC(\LL_\B)$.

\begin{proposition}\label{rc-unary}
Let $S$ be a unary numerical relation and let $S+1:=\{k+1\mid k\in S\}$. Then $\RC(\FO_{\{\le,S\}})\equiv\FO_{\le}(\kC_{S+1})$.
\end{proposition}

\begin{proof}
It is easy to see that $S$ is definable in $\FO_{\le}(\kC_{S+1})$. Hence $\FO_{\{\le,S\}}\le\FO_{\le}(\kC_{S+1})$ by Lemma~\ref{q-semireg}. Furthermore, as noted in Example~\ref{builder}(a), $\FO_{\le}(\kC_{S+1})$ is regular. Thus we see that $\RC(\FO_{\{\le,S\}})\le\FO_{\le}(\kC_{S+1})$.

On the other hand, let $Q$ be the $\BR$-quantifier defined by the $\FO_{\{\le,S\}}[\{U\}]$-sentence $\exists x\,(S(x)\land \forall y\,(U(y)\leftrightarrow y\le x))$.
Then for any $\BR$-model $(\mM,f)$ we have $(\mM,f)\models Q^\reg x,y\,(\phi(x),\psi(y))$ if and only if $|\phi^{\mM,f}|\in S+1$ and $\phi^{\mM,f}$ is an initial segment of $\psi^{\mM,f}$, i.e., $\phi^{\mM,f}$ is of the form $\{a\in \psi^{\mM,f}\mid a\le^f b\}$ for some $b\in \psi^{\mM,f}$. In particular, if $\mM$ is a $\{U\}$-model, we have $(\mM,f)\models\kC_{S+1} x\, U(x)\iff |U^\mM|\in S+1\iff (\mM,f)\models Q^\reg x,y\,(U(x),U(y))$. Thus, the cardinality quantifier $\kC_{S+1}$ is definable in $\RC(\FO_{\{\le,S\}})$, and hence $\FO_{\le}(\kC_{S+1})\le\RC(\FO_{\{\le,S\}})$.
\end{proof}

\subsection*{Regular interior and Crane Beach Conjecture}

The weakness of $\AC$ that we discussed above has earlier inspired
researchers to formulate the so-called {\em Crane Beach
Conjecture}\footnote{Named after the location of an attempt to prove
it.} (CBC) (see \cite{BILST}). The formulation of CBC
is based on the notion of neutral letter. A symbol $e\in\Sigma$ is a
{\em neutral letter} for a language $L\subseteq\Sigma^*$ if for all
$u,v\in\Sigma^*$ it holds that $uv\in L\iff uev\in L$. In other words,
$e$ is a neutral letter for $L$ if inserting or deleting any number of
$e$'s in a word does not affect its membership in $L$.

CBC is the statement that if a language with a neutral letter is
definable in first-order logic with arbitrary built-in relations, then
it is already definable in first-order logic with linear order as the
only built-in relation.
The general form of the conjecture was shown to be false in
\cite{BILST}. However, the paper
\cite{BILST} also provides some interesting restricted
cases in which the conjecture is true.

To formulate these positive results, we will say that a set $\B$ of
built-in relations has the
{\em Neutral Letter Collapse Property\footnote{This notion is similar
to the Crane Beach Property formulated in \cite{LTT}, but not equivalent.}}
(NLCP) with respect to a class $\cal C$ of languages, if the following
holds for every language $L\in{\cal C}$ with a neutral letter:
\begin{quote}
If $Q_L$ is definable in $\FO_\B$, then $Q_L$
is  already definable in $\FO_\le$.
\end{quote}
With this terminology, the relevant  positive results from
\cite{BILST} can be stated as follows:

\begin{theorem}[\cite{BILST}] \label{CBP}
Let $\cal U$ be the set of all unary numerical relations
together with the order $\le$, and let $\cal A$ be the set
of all numerical relations.

\begin{enumerate}[(a)]
\item The set $\cal U$ has NLCP with respect to the class of all languages.
\item The set $\{+\}$ has NLCP with respect to the class of all languages.
\item The set $\cal A$ has NLCP with respect to the class of all languages
in a binary alphabet.
\end{enumerate}
\end{theorem}

We will next show that NLCP can be reformulated in terms of the notion of
regular interior. This is not surprising once we notice that the
property of having a neutral letter is a language theoretic analogue
for the property of a quantifier being universe independent. For
the statement of the result, we need the following concept: given
a language $L\subseteq\Sigma^*$ and a symbol $e\not\in\Sigma$,
define the {\em neutral letter extension} $N(L)$ of $L$ to
be the unique language in the alphabet $\Sigma\cup\{e\}$ such that
$N(L)\cap \Sigma^*=L$ and $e$ is a neutral letter for $N(L)$.

\begin{lemma}
A set $\B$ of built-in relations has NLCP with respect to a class
$\cal C$ of languages if and only if the implication
\[
Q_L \hbox{ is definable in }\RI(\FO_\B) \;\Longrightarrow\;
Q_L \hbox{ is definable in }\FO_\le
\]
holds for every language $L$ such that $N(L)\in{\cal C}$.
\end{lemma}

\begin{proof} 
Assume that $L$ is a language such that $N(L)\in{\cal C}$,
and consider the corresponding language quantifiers $Q_L$ and $Q_{N(L)}$.
It is straightforward to show that $Q_{N(L)}$ is
definable by the regularization $Q_L^\reg$ of $Q_L$, and vice versa.
Moreover, by Lemma~ \ref{regdef}, for any regular $\B$-logic $\LL_\B$,
$Q_L$ is definable in $\LL_\B$
if and only if $Q^\reg_L$ is definable in $\LL_\B$. In particular,
this holds for the logics $\RI(\FO_\B)$ and $\FO_\le$.
Finally, observe that since $Q^\reg$ is universe independent,
it is definable in $\FO_\B$ if and only if it is definable
in $\RI(\FO_\B)$.
The claim follows from these equivalences.
\end{proof}

Note that if $\cal C$ is the class of all languages, then $L\in {\cal C}$
if and only if $N(L)\in {\cal C}$.
Thus, by Theorem \ref{CBP} (a), if $\cal U$ is the set of all unary
arithmetical relations together with the order, $\RI(\FO_{\cal U})$
collapses to $\FO_\le$ if we consider only definability on
word models. Similarly, by Theorem \ref{CBP} (b),
$\RI(\FO_{+})$ collapses to $\FO_\le$ on word models. 
 We
will prove here a stronger result: for both $\FO_{\cal U}$ and $\FO_{+}$,
the regular interior collapses to $\FO_\le$ on all finite models,
not just on word models.

The proof of the collapsing theorem for $\RI(\FO_{\cal U})$
($\RI(\FO_{+})$) is
based on a transfer result stating that
if we are given $\BR$-models $(\mM,f)$ and $(\mN,g)$ which
are $\FO_\le$-equivalent up to large enough quantifier rank,
then we can find paddings $(\mM^*,f^*)$ and $(\mN^*,g^*)$ of
$(\mM,f)$ and $(\mN,g)$ such that $(\mM^*,f^*)$ and $(\mN^*,g^*)$
are $\FO_{\cal U}$-equivalent ($\FO_{+}$-equivalent, respectively)
up to a given quantifier rank $r$.
Here we say that $(\mM^*,f^*)$ is a {\em padding} of $(\mM,f)$ if there is
a set $U\subseteq\Dom(\mM^*)$ such that $(\mM,f)=(\mM^*,f^*)|U$
and $R^{\mM^*}\subseteq U^{\arity(R)}$ for all relation symbols $R$
in the vocabulary of the models. We write $(\mM,f)\pad(\mM^*,f^*)$
if $(\mM^*,f^*)$ is a padding of $(\mM,f)$.

Note that the definition of universe independence
can be restated as invariance with respect to padding: a $\BR$-quantifier
$Q$ is universe independent if and only if the equivalence
$$
   (\mM,f)\in K_Q\iff (\mM^*,f^*)\in K_Q
$$
holds whenever $(\mM,f),(\mM^*,f^*)\in\Str_\BR(\tau_Q)$ and
$(\mM,f)\pad(\mM^*,f^*)$.

We write $(\mM,f)\equiv^r_\B (\mN,g)$ if $(\mM,f)$ and $(\mN,g)$
are $\FO_\B$-equivalent up to quantifier rank $r$, i.e., if
$$
   (\mM,f)\models\phi\iff (\mN,g)\models\phi
$$
for all $\FO_\B$-sentences $\phi$ with quantifier rank at most $r$.
The transfer results can now be formulated as follows:

\begin{lemma}\label{transfer} Let $\cal U$ be as in Theorem \ref{CBP}.
\begin{enumerate}[(a)]
\item (\cite{S3}, Theorem 5.1)
There is a function $i:\NN\to\NN$ such that for all
$\BR$-models $(\mM,f)$ and $(\mN,g)$ with
$(\mM,f)\equiv^{i(r)}_\le (\mN,g)$ there are $\BR$-models $(\mM^*,f^*)$
and $(\mN^*,g^*)$ satisfying $(\mM,f)\pad(\mM^*,f^*)$,
$(\mN,g)\pad(\mN^*,g^*)$ and $(\mM^*,f^*)\equiv^r_{\cal U} (\mN^*,g^*)$.

\item (\cite{S3}, Theorem 6.10) There is 
a function $j:\NN\to\NN$ such that
for all $\BR$-models $(\mM,f)$ and $(\mN,g)$ with
$(\mM,f)\equiv^{j(r)}_\le (\mN,g)$ there are $\BR$-models $(\mM^*,f^*)$
and $(\mN^*,g^*)$ satisfying $(\mM,f)\pad(\mM^*,f^*)$,
$(\mN,g)\pad(\mN^*,g^*)$ and $(\mM^*,f^*)\equiv^r_+ (\mN^*,g^*)$. \qed
\end{enumerate}
\end{lemma}


\begin{theorem}\label{RIcollapse}
Let $\cal U$ be as in Theorem \ref{CBP}.
\begin{enumerate}[(a)]
\item $\RI(\FO_{\cal U})\equiv \FO_\le$.
\item $\RI(\FO_{+})\equiv \FO_\le$.
\end{enumerate}
\end{theorem}

\begin{proof}
(a) Clearly $\FO_\le\le\RI(\FO_{\cal U})$. For the converse
inclusion, it suffices to show
that every universe independent quantifier definable in
$\FO_{\cal U}$ is already definable in $\FO_\le$. Thus, let
$Q$ be a universe independent $\BR$-quantifier, and assume that
its interpreting class
$K_Q\subseteq\Str_\BR(\tau_Q)$ is defined by a sentence
$\phi\in\FO_{\cal U}$ with quantifier rank $r$. We claim that
$K_Q$ is then defined by some sentence $\psi\in\FO_\le$
of quantifier rank $i(r)$, where $i:\NN\to\NN$ is the function
given in Lemma~\ref{transfer}(a).

Assume towards contradiction
that this is not the case. Then there are models $(\mM,f),(\mN,g)
\in\Str_\BR(\tau_Q)$ such that $(\mM,f)\in K_Q$, $(\mN,g)\not\in K_Q$
and $(\mM,f)\equiv^{i(r)}_\le (\mN,g)$. By Lemma~\ref{transfer}(a),
there are paddings $(\mM^*,f^*)$ and $(\mN^*,g^*)$ of
$(\mM,f)$ and $(\mN,g)$, respectively, such that
$(\mM^*,f^*)\equiv^r_{\cal U }(\mN^*,g^*)$. Furthermore, since
the quantifier $Q$ is universe independent, and hence invariant
with respect to padding, we have
$(\mM^*,f^*)\in K_Q$, $(\mN^*,g^*)\not\in K_Q$. This contradicts
the assumption that the defining sentence $\phi\in\FO_{\cal U}$
of $K_Q$ is of quantifier rank $r$.

Item (b) is proved in the same way using Lemma~\ref{transfer}(b).
\end{proof}

It is worth noting that while these results seem to indicate that the
regular interior of $\AC$ is quite weak, a counterexample for the
Crane Beach Conjecture given in \cite{BILST} shows that
it does not entirely collapse to $\FO_\le$. Indeed, the
counterexample shows that there is a language $L$ 
such that $Q_L$ is definable in $\RI(\FO_{\{+,\times\}})$,
but not in $\FO_\le$.
In particular, it is not possible to generalize the third positive
result in Theorem~\ref{CBP} in the same way as we did for the
other two cases in Theorem~\ref{RIcollapse}.

Note however, that the counterexample quantifier $Q_L$ is not order-invariant. 
This is because the language $L$ is not closed under permutations: 
$L=N(L')$, where $L'\subseteq \{0,1,a\}^*$ consists of all words of the form 
$u_0au_1a\ldots au_{2^k-1}$, where $u_0,\ldots,u_{2^k-1}$ lists the
words in $\{0,1\}^k$ in lexicographic order (see Theorem 5.3 in \cite{BILST}).
On the other hand, all the quantifiers corresponding
to languages in a binary alphabet with a neutral letter are
order-invariant; in fact, they are easily seen to be equivalent
with cardinality quantifiers $\kC_S$. 
This raises the question,
whether a correct generalization of Theorem~\ref{CBP}~(c)
would be that all order-invariant language quantifiers definable in
$\RI(\FO_{\cal A})$ are already definable in $\FO_\le$.
In any case, we will show in the next subsection that,
for quantifiers of higher arity, order-invariance
is not a sufficient condition for obtaining a collapse result.

\subsection*{Regular interior of $\AC$}

In this subsection we show that there is an order-invariant
quantifier $Q$
which is definable in $\RI(\FO_{\{+,\times\}})$ but not in $\FO_{\le}$.
For the definition of $Q$, we fix for each
$n>0$ a set $A_n$ such that $|A_n|=n$ and
$A_n\cap \mathcal{P}(A_n)=\emptyset$, and let
$\mA_n\in\Str(\{E\})$ be the structure such that $\Dom(\mA_n)=
A_n\cup (\mathcal{P}(A_n)\setminus\{\emptyset\})$ and
$E^{\mA_n}$ is the membership relation between the sets $A_n$
and $\mathcal{P}(A_n)\setminus\{\emptyset\}$.

\begin{definition} We define $Q$ to be the quantifier with defining class
\begin{equation}\label{powerq}
K_Q=\{ (\mM,f)\in \Str_{\BR}(\{E\}) \mid \exists n( n\in \Sq\setminus\{0\}
\textrm{ and }\mM |U  \cong {\mA_n}) \},
\end{equation}
where $U=\dom(E^{\mM})\cup \rg(E^{\mM})$, and $\Sq=\{n^2\ |\ n\in \NN\}$.
\end{definition}

We will start by showing that
$Q$ can be defined in  $\RI(\FO_{\{+,\times\}})$. Note first that,
for $(\mM,f)\in \Str_{\BR}(\{E\})$, the formulas $\delta(x)\sij
\exists y E(x,y)$
and $\rho(x)\sij \exists y E(y,x)$ define the sets $\dom(E^{\mM})$ and
$\rg(E^{\mM})$ in $(\mM,f)$, respectively. Let $\chi$ be the conjunction
of the following $\FO(\{E\})$-sentences:
\begin{itemize}
\item[] $\alpha:=\exists x\,\delta(x)\land\forall x (\delta(x)\to\lnot\rho(x))$
\item[] $\beta:=\forall x\forall y((\rho(x)\wedge \rho(y)\wedge
\forall z (E(z,x) \leftrightarrow E(z,y)))\to x=y)$
\item[] $\gamma:=\forall x\forall y(\delta(x)\to
\exists z\forall w
(E(w,z)\leftrightarrow (w=x\lor E(w,y)))$.
\end{itemize}
Note that $\alpha$ says that $\dom(E)\not=\emptyset$ and $\dom(E)\cap\rg(E)=\emptyset$,
and $\beta$ says that $E$ is extensional. Furthermore, $\gamma$ says that given any
element $x\in\dom(E)$ and subset $E_y=\{w\mid (w,y)\in E\}$ of $\dom(E)$,
the subset $E_y\cup\{x\}$ is also of the form $E_z=\{w\mid (w,z)\in E\}$ for some $z\in \rg(E)$.
In particular, if there is an element $y\not\in\rg(E)$, then $\gamma$ implies that
every singleton subset $\{x\}$, and by induction, every non-empty subset  $X$ of $\dom(E)$ is of the 
form $E_z$ for some $z\in\rg(E)$.

Using the observations above, it is straightforward to verify that the equivalence
\[
    (\mM,f)\models \chi\iff \mM|U \cong {\mA_n}
    \hbox{ for some }n>0
\]
holds for all  $(\mM,f)\in \Str_{\BR}(\{E\})$. We still need
to  express the condition
$n\in \Sq$ in the definition of $K_Q$. To do this, we apply the
so-called  polylogarithmic counting ability of
$\FO_{\{+,\times\}}$:

\begin{theorem}[\cite{AB,FKPS,DGS,WWY}] \label{counting}
The logic $\FO_{\{+,\times\}}$ can
count up to $\lb^k$ for any $k\in\NN$. In other words, for every $k$,
there is a formula $\sigma_k(x)\in \FO_{\{+,\times\}}[\{P\}]$, where
$P$ is unary, such that for all $(\mM,f)\in \Str_{\BR}(\{P\})$
\[
    (\mM,f)\models \sigma_k[a/x]\iff |\{b\in\Dom(\mM)\mid b<^f a\}|=
    |P|\le \lb^k n,
\]
where $n=|\Dom(\mM)|$.
\end{theorem}

Let $\theta(x)$ be the $\FO_{\{+,\times\}}[\{E\}]$-formula obtained
from $\sigma_1(x)$ by substituting the formula $\delta$ in place of the
relation symbol $P$. Furthermore, let $\psi$ be the
$\FO_{\{+,\times\}}[\{E\}]$-sentence
$\exists x (\theta(x)\land \exists y(x=y\times y))$.
By Theorem \ref{counting}, $(\mM,f)\models\psi$ if and only if
$|\dom(E^\mM)|\in\Sq$ and $|\dom(E^\mM)|\le \lb (|\Dom(\mM)|)$.
Note that the latter condition is automatically true if $(\mM,f)\models\chi$.
Thus, we conclude that the quantifier $Q$ is  defined
by the $\FO_{\{+,\times\}}[\{E\}]$-sentence $\chi \land \psi$.
Finally, observe that
the quantifier $Q$ is universe independent, whence it is contained
in $\RI(\FO_{\{+,\times\}})$.

Next we show that $Q$ is not definable in $\FO_{\le}$. Towards a
contradiction, assume that  $Q$ can be defined in  $\FO_{\le}$
by a sentence $\eta$. We will show that then the language
$L=\{ w\in\{a\}^*\mid |w|\in \Sq \}$ is definable in $\MSO_{\le}$
contradicting  the fact that all $\MSO$-definable  languages
are regular and $L$ is not.

For each $n>0$, let $A_n$ and $\mA_n$ be as above, and let
$(\mM_n,f_n)$ be the $\Str_\BR(\{E\})$-structure such that
$\Dom(\mM_n)=\Dom(\mA_n)$, $E^{\mM_n}=E^{\mA_n}$ and
the ordering $\le^{f_n}$ satisfies the condition:
\begin{itemize}
\item ${\le^{f_n}}\cap (\mathcal{P}(A_n)\setminus\{\emptyset\})^2$
is the lexicographic
ordering of subsets of $A_n$ induced by ${\le^{f_n}}\cap A_n^2$,
\item   $a\le^{f_n} b$ for all $a\in A_n$ and
$b\in(\mathcal{P}(A_n)\setminus\{\emptyset\})$.
\end{itemize}
Furthermore, for each $n>0$, we let $(\mN_n,g_n)\in\Str_\BR(\{P_a\})$ 
be the
word model of length $n$ with $\Dom(\mN_n)=P_a^{\mN_n}=A_n$ and
${\le^{g_n}}={\le^{f_n}}\cap A_n^2$.

We will show that any $\FO_{\le}[\{E\}]$-sentence
$\phi$ can be translated into a sentence $\phi^*\in
\MSO_{\le}[\{ P_a\}]$ such that for all $n$
\[
(\mM_n,f_n)\models \phi \iff (\mN_n,g_n)\models \phi^*.
\]
By the assumption, we have $(\mM_n,f_n)\models \eta$ if and only if
$|A_n|\in \Sq$. Therefore, $\eta^*$ will then define the
language $L=\{ w\in\{a\}^*\mid |w|\in \Sq \}$. This will be the desired
contradiction.

The idea of the translation $\phi\mapsto \phi^*$ is simple: each
element $a\in\Dom(\mM_n)\setminus\Dom(\mN_n)$ is also a
subset of $\Dom(\mN_n)$. Thus, first-order variable and quantification
over elements in $\Dom(\mM_n)\setminus\Dom(\mN_n)$ can be replaced
by monadic second-order variables and quantification.

We will now describe the technical details of the translation.
First, we assign to each first-order variable a corresponding capitalized
monadic second-order variable. 
For translating formulas
with free variables, we need to keep track of those
variables which should be translated into the corresponding
second-order variables. Thus, we define
a translation $T_S: \FO_{\le}[\{E\}]\to \MSO_{\le}[\{P_a\}]$
for each set $S$ of first-order variables by simultaneous induction:

\begin{align*}
T_S(x=y) &\sij
\begin{cases} x=y & \text{if $x,y\not\in S$}
\\
\forall z(X(z)\leftrightarrow Y(z)) & \text{if $x,y\in S$}
\\
x\not=x & \text{otherwise}
\end{cases}
\\
T_S(x\le y) &\sij
\begin{cases} x\le y & \text{if $x,y\not\in S$}
\\
x=x & \text{if $x\not\in S$, $y\in S$}
\\
x\not=x & \text{if $x\in S$, $y\not\in S$}
\\
X\le Y & \text{if $x,y\in S$}
\end{cases}
\\
T_S(E(x,y)) &\sij
\begin{cases} Y(x) & \text{if $x\not\in S$, $y\in S$}
\\
x\not=x & \text{otherwise}
\end{cases}
\\
T_S(\lnot \phi)  &\sij  \lnot T_S(\phi)
\\
T_S(\phi \land \psi)   &\sij   T_S(\phi) \land T_S(\psi)
\\
T_S(\exists x\,\phi) &\sij  \exists x\, T_S(\psi)\lor
\exists X\, T_{S\cup\{x\}}(\psi) \\
\end{align*}
Above, $ X\le Y$ denotes the formula which defines the lexicographic
ordering of subsets induced by $\le$.

By induction on the construction of $\phi\in \FO_{\le}[\{E\}]$ one can
now prove that, for all $n>0$, and $\va=(a_1,\ldots, a_{k})
\in A^k_n$ and  $\vb=(b_1,\ldots, b_{l})\in
(\mathcal{P}(A_n)\setminus\{\emptyset\})^l$ it holds that
\begin{equation*}\label{translation}
(\mM_n,f_n) \models \phi[\va/ \vx, \vb/\vy ] \iff (\mN_n,g_n)\models
T_S(\phi)[\va/ \vx, \vb/\vY ]
\end{equation*}
whenever $x_i\not\in S$ for each component $x_i$ of $\vx$ and
$y_j\in S$ for each component $y_j$ of $\vy$.
In particular, if $\phi$ is a sentence, we have
$(\mM_n,f_n) \models \phi\iff (\mN_n,g_n)\models T_{\emptyset}(\phi)$. Thus,
defining $\phi^*\sij T_\emptyset(\phi)$ we get the desired translation.

\subsection*{Regular closure of $\AC$}

The gap between $\RI(\FO_\B)$ and $\RC(\FO_\B)$ can be seen as
a measure for the irregularity of $\FO_\B$: the larger the gap is,
the more irregular $\FO_\B$ is. We will next show that in the case
of $\B$ consisting of a suitable unary relation and the order $\le$,
this gap is extremely large. We have already seen in Theorem \ref{RIcollapse}  (a) that for such a
$\B$, $\RI(\FO_\B)\equiv\FO_\le$. 
For the other direction, we have

\begin{theorem}\label{RCTC}
Let $\B$ be 
a set of built-in relations such that $\B$ contains the order
$\le$ and
a pseudoloose set~$S$. Then $\RC(\FO_\B)\equiv\FO_{\B}(\kI)$. Moreover,
$\TC\le\RC(\FO_{\B})$.
\end{theorem}

\begin{proof}
By Proposition~\ref{rc-unary},
$\RC(\FO_{\{\le,S\}})\equiv\FO_{\le}(\kC_{S+1})$, where $S+1=\{k+1\in\NN\mid k\in S\}$. Since $\le$ and  $S$ are in $\B$, we have $\FO_{\le}(\kC_{S+1})\le\RC(\FO_\B)$. Furthermore, since $S$ is pseudoloose, $S+1$ is pseudoloose as well (cf.~Proposition~\ref{loose-biLip}), and hence by Proposition 3.6,
the H\"artig quantifier
$\kI$ is definable in $\FO_{\le}(\kC_{S+1})$. This shows that
$\FO_\B(\kI)\le\RC(\FO_\B)$.

On the other hand, as we showed in
Example~\ref{builder},
$\FO_\B(\kI)$ is always regular, whence it necessarily contains
the least regular extension $\RC(\FO_\B)$ of $\FO_{\B}(\kI)$. Thus, we see that $\RC(\FO_\B)\equiv\FO_{\B}(\kI)$.
Finally,
since $\kC_{S+1}$ is definable in $\RC(\FO_\B)$, the second claim
follows directly from Theorem~\ref{tcchar}.
%
%
\end{proof}

The proof above makes use of quantifiers with {\em empty vocabulary}. 
Syntactically, such a quantifier is considered as a sentence $Q$
which is in $\LL_\B(Q)[\tau]$ for every vocabulary~$\tau$.  The semantics
of $Q$ is always defined by some set $S\subseteq\NN$:  
$(\mM,f)\models Q$ if and only if $|\Dom(\mM)|\in S$. However,
the reference to quantifiers with empty vocabulary can be easily avoided:
instead of the quantifier defined in the proof, we could use the
quantifier $Q^*$ with defining class 
$$
  K_{Q^*}=\{(\mM,f)\in\Str_\BR(U)\mid U=\Dom(\mM)\hbox{ and } |U|-1\in S\}.
$$
This is because 
$Q^\reg$ 
and $Q^*$ are obviously definable from each
other.
As an immediate corollary for Theorem~\ref{RCTC}, we get that the 
regular closure of
$\FO_{\{+,\times\}}$ is equal to $\TC$. This is because for example
the pseudoloose set $\Sq=\{n^2\mid n\in\NN\}$ is definable in
$\FO_{\{+,\times\}}$. We formulate the result in terms of the circuit
complexity classes:

\begin{corollary}
$\RC(\AC)\equiv\TC$. \qed
\end{corollary}

We will close this section by some examples illustrating the
gap between regular interior and regular closure.
All of our non-regular examples have the same regular interior $\FO_{\le}$,
but the regular closure varies (see also the related figure~\ref{rgap}).

\begin{figure}
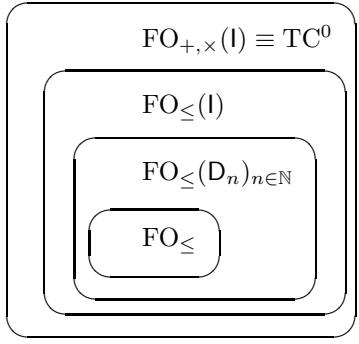
 
    $$\xy
0;<15mm,0mm>:
<-4mm,-2.5mm> !LU+<-2mm,4mm>="1",
0*+<2mm>!LU{\FO_{\le}} !RD+<2mm,-2mm>="1t",
"1"."1t" *\frm<3mm>{-},
!LD+<-2mm,-1mm>="1b" !LD+<-4mm,-2mm>="1c" !LD+<-5mm,-3mm>="1d";
(0,0.6)*+<2mm>!LU{\FO_{\le}(\kD_n)_{n\in\NN}},
!RU+<2mm,2mm>="2t", !RU+<4mm,4mm>="2u",
"1b"+<0mm,-2mm>."2t" *\frm<3mm>{-};
(0,1.2)*+<2mm>!LU{\FO_{\le}(\kI)}
!RU+<2mm,2mm>="3",
"1c"+<0mm,-2mm>."2u"."3" *\frm<3mm>{-};
(0,1.8)*+<2mm>!LU{\FO_{{+},{\times}}(\kI)\equiv\TC}
!RU+<2mm,2mm>="3",
"1d"+<0mm,-2mm>."2u"."3" *\frm<3mm>{-}
    \endxy$$
\caption{Some regular logics with built-in relations}
\label{rgap}
\end{figure}
\begin{example}\label{reggap}
\begin{enumerate}[(a)]
\item
Let $S=\rg(P)$ for some polynomial with integer coefficients
and degree at least~$2$. Then Theorem~\ref{RIcollapse}(a) implies that
$\RI(\FO_{\{\le,S\}})\equiv\FO_\le$, and Theorem~\ref{RCTC} implies that
$\RC(\FO_{\{\le,S\}})\equiv\FO_{\{+,\times\}}(\kI)\equiv\TC$.

\item 
By Theorem~\ref{RIcollapse}(b), $\RI(\FO_+)\equiv\FO_\le$.
On the other hand, 
we will show that
$\RC(\FO_+)\equiv\FO_\le(\kI)$.  
Let us consider the language $L_0=\joukko{a^nb^n \mid n\in\NN}$.
Then the corresponding language quantifier $Q_{L_0}$ is definable in~$\FO_+$ and $Q^{\reg}_{L_0}$ in $\RC(\FO_+)$.
Observing that the extra predicate $P$ in the vocabulary of $Q^{\reg}_{L_0}$
is futile, we get that $\FO_\le(Q^{\reg}_{L_0})\equiv\FO_\le(Q_{L_1})$ where
$L_1=N(L_0)\subseteq\joukko{a,b,e}^*$ is the neutral letter extension of $L_0$.

Let us show that $\kI$ is definable in $\FO_\le(Q_{L_1})$.
Let $(\mM,f)\in\Str_\BR(\joukko{U,V})$. 
Note first that the
parity of $U^\mM$ is expressible in $\FO_\le(Q_{L_1})$: $|U^\mM|$ is even if and only if $(\mM,f)\models\exists z\,\phi(z)$, where $\phi(z):=U(z)\land Q_{L_1}\,x,y\,(U(x)\land x\le z,U(y)\land z<y)$.  
Furthermore, the median $u$ of the set $U^\mM$ (i.e., the unique element $u$ of $U^\mM$ such that $|\{a\in U^\mM\mid a\le^f u\}|=\lfloor\frac{1}{2}|U^\mM|\rfloor$) is definable in $\FO_\le(Q_{L_1})$: $u$ is the unique element of $\Dom(\mM)$ such that $(\mM,f)\models \phi'[u/z])$, where $\phi'=\phi$ if $|U^\mM|$ is even and $\phi'(z):=U(z)\land Q_{L_1}\,x,y\,(U(x)\land x< z,U(y)\land z<y)$ if $|U^\mM|$ is odd. In the same way we can define the parity of $|V^\mM|$ by a $\FO_\le(Q_{L_1})$-sentence $\exists z\,\psi(z)$ and the median $v$ of $V^\mM$ by a $\FO_\le(Q_{L_1})$-formula $\psi'(z)$.

We can now express
$|U^\mM|=|V^\mM|$ as a disjunction of cases according to the
parities of $U^\mM$ and $V^\mM$ and the relative locations of
the medians $u$ and $v$. Note first that if $|U^\mM|$ is even and $|V^\mM|$ is odd, or vice versa, then $|U^\mM|=|V^\mM|$ is trivially false, whence we only need to consider four cases based on the shared parity of $|U^\mM|$ and $|V^\mM|$ and the truth value of $u\le^f v$.

Suppose that $|U^\mM|$ and $|V^\mM|$
are both odd and $u\le^f v$.  The point $u$ splits $U^\mM$ in
intervals in the following way: $U^\mM=U_-\cup\joukko{u}\cup U_+$
where $a<^f u <^f b$, for each $a\in U_-$ and $b\in U_+$,
and $|U_-|=|U_+|$.
Similarly, $v$ splits $V^\mM$ as $V^\mM=V_-\cup\joukko{v}\cup V_+$.
Then $|U^\mM|=|V^\mM|$ is equivalent to $|U_-|=|V_+|$ and the latter
is directly expressible in $\FO_\le(Q_{L_1})$, since
$a <^f  u \le^f v <^f b$, for all $a\in U_-$ and $b\in V_+$:
$$
    |U_-|=|V_+|\iff (\mM,f)\models Q_{L_1}\,x,y\,(U(x)\land x< z,V(y)\land z<y)[u/z,v/z'].
$$
Thus, we see that the disjunct 
$$
    \exists z\exists z'\,(\phi'(z)\land\psi'(z')\land Q_{L_1}\,x,y\,(U(x)\land x< z,V(y)\land z<y))
$$
covers the case where $|U^\mM|$ and $|V^\mM|$ are odd and $u\le^f v$. The other cases are similar.
Hence, we get $\FO_+\le\FO_\le(\kI)\le\FO_\le(Q_{L_1})\le\RC(\FO_+)$, and we conclude that $\RC(\FO_+)\equiv\FO_\le(\kI)$, because $\FO_\le(\kI)$
is regular.

\item
Let $S_n=n\NN=\{nk\mid k\in\NN\}$ for each $n\in\NN$.
Again, by Theorem~\ref{RIcollapse}(a), $\RI(\FO_{\{\le,S_n\}})
\equiv\FO_\le$. For the other direction we have
$\RC(\FO_{\{\le,S_n\}})\equiv\FO_\le(\kC_{S_n+1})$ by Proposition~\ref{rc-unary}. Note that the formula $\kC_{S_n+1}\,x\, U(x)$ just says that $|U^\mM|\equiv 1 \pmod{n}$. Using this it is easy to verify that $\kC_{S_n+1}$ and the divisibility quantifier $\kD_n$ are definable from each other. Thus we see that actually $\RC(\FO_{\{\le,S_n\}})\equiv\FO_\le(\kD_n)$.

\item
Let $\B=\{\le\}\cup\{S_n\mid n\in\NN\}$. Applying Theorem~\ref{RIcollapse}(a) again we see that $\RI(\FO_\B)\equiv\FO_\le$. By generalizing Proposition~\ref{rc-unary} to a set of unary built-in relations and using the observation in the previous example, we see that $\RC(\FO_\B)\equiv\FO_\le(\kD_n)_{n\in\NN}$.
\end{enumerate}
\end{example}

Note that while the regular interior of each logic considered in the example above is $\FO_\le$, their regular closures form a strictly increasing sequence: 
$$
    \FO_\le(\kD_n) < \FO_\le(\kD_n)_{n\in\NN} < \FO_\le(\kI) < \FO_{\{+,\times\}}(\kI).
$$
Here the strictness of the first containment $\FO_\le(\kD_n) \le\FO_\le(\kD_n)_{n\in\NN}$ follows from a result in \cite{N} stating that $\FO_\le(\kD_n)\le\FO_\le(\kD_m)$ if and only if all prime factors of $n$ are also prime factors of $m$.
The strictness of the second containment $\FO_\le(\kD_n)_{n\in\NN} \le \FO_\le(\kI)$ also follows from this, since otherwise $\kI$ would be definable in terms of a finite set $\mathcal{D}=\{\kD_{n_1},\ldots,\kD_{n_k}\}$ of divisibility quantifiers, contradicting the fact that $\kD_p$ is not definable in $\FO_\le(\mathcal{D})$ for a prime $p\ge\max\{n_1,\ldots,n_k\}$. 
Finally, it was observed in \cite{Lu1} that $\FO_\le(\kI)\equiv\FO_\le(\Maj)$, where $\Maj$ is the unary majority quantifier (see Section~\ref{Background}), and $\FO_\le(\Maj)$ was proved in \cite{Li} to be strictly weaker than $\FO_\le(\Maj^2)$ that captures $\TC\equiv\FO_{\{+,\times\}}(\kI)$.

\section*{Conclusion}

In this paper, the circuit complexity class $\TC$ has been
studied from the perspective of descriptive complexity theory.
It is natural to ask what kind of directions and
open questions our approach offers for future research. 
Finite model theory is concerned about a multitude of different logics.
In this context, the complexity classes serve as important milestones,
as all of the most important complexity classes have been characterized
by various logics.  However, it is notoriously well known that
we usually do not even know if the milestones are different, even if
we believe that they are far apart, and descriptive complexity
has not (yet) helped much here, except for the (non-separation)
result of Immerman and Szelepcs{\'e}nyi \cite{I2,Sz}.

As it stands, a more granulated picture could help:  instead of just
studying complexity classes one can study other logics which bear
some similarities with complexity classes. Since any semiregular
logic can be represented in terms of quantifiers, it is quite natural
to search for characterizations of complexity classes by
quantifier logics, and compare them with such logics.  To take
an example, consider the separation $\TC<\PSPACE$.  We know that
$\TC\equiv\FO_\le(\kC_\Sq)$, but by the results of Dawar and Hella
\cite{DH}, $\PSPACE\not\equiv\FO_\le(Q)$ for any single
quantifier~$Q$.  Hence, the separation is manifested also in distinctive
properties of these complexity classes, one is generated by a single
quantifier on ordered structures, the other is not.

This raises questions for complexity classes in between:
Does $\PTIME$ have a representation of the form $\PTIME\equiv\FO_\le(Q)$?
If it does, $\PTIME<\PSPACE$, if it does not, then $\TC<\PTIME$.
It is known that 
$\PTIME\equiv\LFP_\le\equiv\FO_\le(\kQ^{<\omega}_{\mathrm{ATC}})$,
but this is a representation in terms of a sequence of quantifiers
(vectorizations of $\kQ_{\mathrm{ATC}}$), rather than
in terms of a single quantifier.  Obviously, it may also be asked
if any of $\NP$, $\LOGSPACE$ or $\NLOGSPACE$ is generated by
a single quantifier.  For $\NP$, the solution of this problem
would also solve open problems in complexity theory.

Quantifier representations enable a finer analysis of logics, in our case
the circuit complexity class~$\TC$.  In Section~\ref{card} we arrived
at the notion of pseudolooseness via natural combinatorial considerations
when we tried to find sufficient conditions for $S\osaj\NN$ such
that $\FO_\le(\kC_S)\ge\TC$.  We dare not anticipate that this condition
would also be necessary, since it is likely that there is some
technical finesse that has been overlooked.  The real question
is rather if there exist sets~$S$ that are very far from being
pseudoloose, but satisfy $\FO_\le(\kC_S)\ge\TC$.  To be concrete,
we ask if $E=\{2^n\ |\ n\in \NN \}$ is such a set. 
Note that $E$ is not pseudoloose by Example~\ref{nonpseudoloose}
 but by~(\ref{Idef-E})
of Lemma~\ref{Idefex} it holds that
   \[
\FO_{\le}(\kC_{E})\equiv \FO_+(\kI,\kC_{E}) 
\le \FO_{\{+,\times\}}(\Maj).
   \] 
It is an open question whether the latter two logics are equivalent. 
Another interesting concrete case is the same question for the
set of primes $\PP$.
More generally, we ask if there is $S\subseteq \NN$ such
that
   \[ 
\FO_{\le}(\kI)<  \FO_{\le}(\kC_S)<\FO_{\le}(\kD)
\equiv \FO_{\{+,\times\}}(\Maj).  
   \]

Section~\ref{card} reveals an interesting methodological point:
When concrete examples of pseudoloose sets were sought,
some elementary analytic methods were employed.  One should
compare this with number theory, where analytic methods are
well established, especially in the context of studying the set~$\PP$.
It is conceivable that the problems in complexity theory are comparable
to those in number theory. Could analytic methods be helpful in
broader generality in finite model theory?  In other words,
could we see the rise of \emph{analytic finite model theory?}

In Sections~\ref{Builtin} and~\ref{CraneBeach}  the point of view is slightly changed.  Speaking in general terms, we are concerned
about what closure properties logics characterizing complexity classes
have.  For the nondeterministic case, the major question is if
they are closed under negation, whereas for deterministic complexity
classes, the corresponding logics have usually quite good
closure properties.  Our remark is that $\TC$ is regular, whereas
$\AC$ is not.  Analyzing this a bit further, we studied the notions
of  regular closure and interior.  We summarize some of the results
in the table below.

\begin{center}
\begin{tabular*}{8cm}{lll}
Logic & $\RI$  & $\RC$ \\
\hline
$\FO_\le$                  & $\FO_\le$  & $\FO_\le$  \\
$\FO_{\joukko{\le,n\NN}}$  & $\FO_\le$  & $\FO_\le(\kD_n)$  \\
$\FO_+$                    & $\FO_\le$  & $\FO_\le(\kI)$  \\
$\FO_{\joukko{\le,\Sq}}$   & $\FO_\le$  & $\FO_\le(\kC_\Sq)\equiv\TC$  \\
$\FO_\bit\equiv\AC$        & $>\FO_\le$ & $\TC$ \\
\end{tabular*}
\end{center}

There is one incomplete entry in this table: we do not know
any nice explicit description for the regular interior of $\AC$. 
In particular,
it is an open problem
whether there exists a single quantifier $Q$
and a set $\B$ of numerical relations such that 
$\RI(\AC)\equiv\FO_\B(Q)$.

The last question concerns the counterexample  
\cite{BILST} for the Crane Beach Conjecture showing
that the regular interior of $\AC$ does not entirely collapse to $\FO_\le$. This
counterexample is a language which is not order-invariant. In fact, as
far as we know, it is an open question whether the order-invariant
version of NLCP holds for $\AC$ with respect to the class of all languages, i.e., is every order-invariant
$\FO_{\{+,\times\}}$-definable language with a neutral letter already
definable in $\FO_{\le}$. This question can be equivalently formulated
as follows: is every $\FO_{\{+,\times\}}$-definable unary quantifier
$Q$ already definable in $\FO_\le$.

\bibliographystyle{alpha}
\bibliography{rrutc}

\end{document}